\documentclass[onecolumn,12pt]{IEEEtran}

\usepackage[T1]{fontenc}
\usepackage{tikz}
\usetikzlibrary{arrows}
\usetikzlibrary{calc}
\usetikzlibrary{snakes}
\usepackage{cite}
\usepackage{graphicx}
\usepackage{array}
\usepackage{mdwmath}
\usepackage{mdwtab}
\usepackage{eqparbox}
\usepackage{fixltx2e}
\usepackage{url}
\usepackage{comment}
\usepackage{subcaption}
\usepackage{bbm}

\usepackage{amsmath,amsfonts,amssymb,amsthm}
\usepackage{xspace}
\usepackage{xcolor}
\usepackage{algorithmic}
\usepackage{bm}

\usepackage{booktabs}
\usepackage{nicefrac}
\usepackage{textcomp}
\usepackage{stmaryrd}
\usepackage{mathrsfs}
\usepackage{paralist}

\usepackage[font=footnotesize,labelsep=space,labelfont=bf]{caption}
\usepackage{url}

\usepackage[pdftex,bookmarks=true,pdfstartview=FitH,colorlinks,linkcolor=blue,filecolor=blue,citecolor=blue,urlcolor=blue]{hyperref}
\makeatletter

\newcommand{\Rmnum}[1]{\expandafter\@slowromancap\romannumeral #1@}
\makeatother

\newtheorem{thm}{Theorem}
\newtheorem{defn}{Definition}
\newtheorem{lem}{Lemma}
\newtheorem{claim}{Claim}

\newtheorem{remark}{Remark}

\newcommand{\namedref}[2]{\hyperref[#2]{#1~\ref*{#2}}}

\newcommand{\Sectionref}[1]{\namedref{Section}{sec:#1}}
\newcommand{\Subsectionref}[1]{\namedref{Subsection}{subsec:#1}}

\newcommand{\Appendixref}[1]{\namedref{Appendix}{app:#1}}
\newcommand{\Theoremref}[1]{\namedref{Theorem}{thm:#1}}

\newcommand{\Lemmaref}[1]{\namedref{Lemma}{lem:#1}}

\newcommand{\Claimref}[1]{\namedref{Claim}{claim:#1}}

\newcommand{\Figureref}[1]{\namedref{Figure}{fig:#1}}

\newcommand{\Pageref}[1]{\hyperref[#1]{page~\pageref*{#1}}}

\definecolor{darkred}{rgb}{0.5, 0, 0} 
\definecolor{darkgreen}{rgb}{0, 0.5, 0} 
\definecolor{darkblue}{rgb}{0,0,0.5} 
\hypersetup{
    colorlinks=true,
    linkcolor=darkred,
    citecolor=darkblue,
    urlcolor=darkblue   
}

\newcommand{\B}{\ensuremath{\mathcal{B}}\xspace}
\newcommand{\X}{\ensuremath{\mathcal{X}}\xspace}
\newcommand{\Y}{\ensuremath{\mathcal{Y}}\xspace}
\newcommand{\Z}{\ensuremath{\mathcal{Z}}\xspace}

\newcommand{\U}{\ensuremath{\mathcal{U}}\xspace}
\newcommand{\V}{\ensuremath{\mathcal{V}}\xspace}

\newcommand{\C}{\ensuremath{\mathcal{C}}\xspace}

\renewcommand{\paragraph}[1]{\smallskip\noindent{\bf #1}~}

\begin{document}

\title{Secure Computation of Randomized Functions: Further Results}
\author{\IEEEauthorblockN{Deepesh Data and Vinod M. Prabhakaran} \\
\IEEEauthorblockA{School of Technology \& Computer Science\\
Tata Institute of Fundamental Research,
Mumbai, India\\
Email: \{deepeshd,vinodmp\}@tifr.res.in}} 
\maketitle
\pagestyle{plain}

\begin{abstract}
We consider secure computation of randomized functions between two users, where both the users (Alice and Bob) have inputs, Alice sends a message to Bob over a rate-limited, noise-free link, and then Bob produces the output. We study two cases: (i) when privacy condition is required only against Bob, who tries to learn more about Alice's input from the message than what can be inferred by his own input and output, and (ii) when there is no privacy requirement. For both the problems, we give single-letter expressions for the optimal rates. For the first problem, we also explicitly characterize securely computable randomized functions when input has full support, which leads to a much simpler expression for the optimal rate. Recently, Data (ISIT 2016) studied the remaining two cases (first, when privacy conditions are against both the users; and second, when privacy condition is only against Alice) and obtained single-letter expressions for optimal rates in both the scenarios.
\end{abstract}

\section{Introduction}\label{introduction}
Two-user secure computation allows mutually distrusting users to jointly perform computation of their private data interactively, in such a way that no individual learns anything beyond the function value. 
The study of secure computation was initiated in \cite{ShamirRiAd79}, \cite{Rabin81}, \cite{Yao82,Yao86}, etc., under computational assumptions, and culminated in a surprising result that, every function (both deterministic and randomized) can be securely computed.
However, it turns out that not all two-user functions are computable with information-theoretic security. A combinatorial characterization of securely computable two-user {\em deterministic} functions (with perfect security), along-with a generic round-optimal secure protocol, was given in \cite{Kushilevitz89}, \cite{Beaver89b}. Maji et al. \cite{MajiPrRo09}, among other things, characterized deterministic functions that have (statistically) secure protocols.
An alternative characterization (for perfect security) was given by Narayan et al. \cite{NarayanTyWa15} using common randomness generated by deterministic secure protocols. 
A characterization of securely sampleable joint distributions (in terms of Wyner commom information \cite{Wyner75}) was given in \cite{WangIs11}, and of securely computable functions, when communication is public and privacy of the function value is against an eavesdropper, was given in \cite{TyagiNaGu11}.

Interestingly, characterization for securely computable {\em randomized} functions is still not known. Kilian \cite{Kilian00}, among many other things, gave a characterization when both the users have inputs and only Bob produces the output. Data \cite{Data16} studied two variations of the same problem (one with privacy against both the users, and other with privacy only against Alice) with a probability distribution on inputs and gave rate-optimal codes for both the problems in perfect security as well as in asymptotic security settings. In this paper we study the remaining two cases (one with privacy only against Bob, and the other with no privacy), and give rate-optimal codes for both the problems in asymptotic security setting, when only one-sided communication from Alice to Bob is allowed. Proofs in this paper differ in significant ways from that of \cite{Data16}. We also give explicit characterization of securely computable randomized functions (where inputs have full support), for the case when privacy is only against Bob, which leads to a much simpler expression for the optimal rate. Note that Ortiksky and Roche \cite{OrlitskyRoche} studied the same problem for {\em deterministic} function computation with no privacy and gave a rate-optimal code, but their proofs do not seem to generalize to the randomized function setting that we consider in this paper.

\begin{figure}
\centering
\begin{tikzpicture}[>=stealth']
\draw [fill=lightgray, thick] (0,0) rectangle node {\Large A} +(1,1); \draw [->,thick] (-0.65,0.5) -- (0,0.5); \node at (-0.9,0.6) {$X^n$};
\draw [fill=lightgray, thick] (4,0) rectangle node {\Large B} +(1,1); \draw [<-,thick] (5,0.5) -- (5.65,0.5); \node at (6,0.55) {$Y^n$};
\draw [->,thick] (4.5,0) -- (4.5,-0.5); \node at (4.5,-0.8) {$\hat{Z}^n$}; \node [right,scale=0.8] at (5,-0.5) {$Z_i\sim p_{Z|XY}$};
\draw [->,thick] (1,0.5) -- (4,0.5); \node at (2.5,0.75) {$M$};
\end{tikzpicture}
\caption{
A secure computation problem is described by $(p_{XY},p_{Z|XY})$: Alice and Bob get $X^n$ and $Y^n$, respectively, as their inputs, where $(X_i,Y_i)$'s are i.i.d. according to $p_{XY}$, and Bob wants to securely compute $Z^n$ according to $p_{Z^n|X^nY^n}(z^n|x^n,y^n)=\Pi_{i=1}^np_{Z|XY}(z_i|x_i,y_i)$. Users have access to private randomness. In any secure code, Alice sends a message $M$ over a rate-limited, noise-free link to Bob, and based on $(M,Y^n)$ Bob produces $\hat{Z}^n$ as his output, such that the following conditions are satisfied when block-length tends to infinity: (i) the $L_1$-distance between the induced distribution $p_{X^nY^n\hat{Z}^n}$ and the desired distribution $p_{X^nY^nZ^n}$ goes to zero, and (ii) the average amount of additional information Bob learns about $X^n$ from the message $M$ goes to zero. We also study this problem when there is no privacy, i.e., purely from the (randomized) function computation point of view. There are no external adversaries/eavesdroppers.}
\label{fig:problem}
\end{figure}
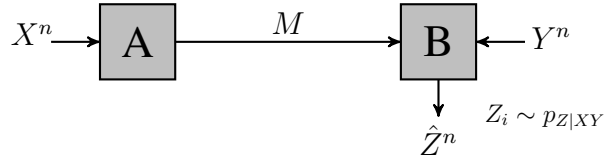
The techniques used in our achievability proofs were developed by Yassaee, Aref, and Gohari \cite{YassaeeArGo14}, who, along with several new results, also gave alternative and, arguably, simpler proofs for many well-studied problems in network information theory, including channel synthesis problems. Cuff \cite{Cuff13} used a different tool (soft-covering lemma) to synthesize a memoryless channel with the help of communication and common randomness. Our problem can be viewed as a generalization of channel synthesis problem with privacy. Both \cite{YassaeeArGo14} and \cite{Cuff13} studied several security related problems (secret key generation, wiretap channels, etc.), but all of them guarantee privacy against an eavesdropper. The focus of this paper is on secure computation (where there is no eavesdropper and we want privacy against the users themselves), which is fundamentally different from secure communication or secret key agreement (where we want privacy against an eavesdropper).

This paper is organized as follows: in \Sectionref{problem-defn} we formally define two problems (one with privacy, and the other one without privacy) and state our main theorems for the optimal rate-expressions. 
In \Sectionref{proof_general-case} we prove the first theorem (one with privacy). We also explicitly characterize securely computable randomized functions when inputs have full support, which leads to a more explicit expression for the optimal rate. Using this characterization we give a direct achievability using the Slepian-Wolf coding scheme, along with a much simpler argument of its correctness and privacy. 

\paragraph{Notation.} 
For $n\in\mathbb{N}$, we write $[n]$ to denote the set $\{1,2,\hdots,n\}$.
We write capital letters $P, Q,$ etc., to denote random p.m.f.'s and small letters $p,q,$ etc., to denote non-random p.m.f.'s.
The total variation distance between two p.m.f.'s $p_U$ and $q_U$ on the same alphabet $\U$ is defined by $\|p_U-q_U\|_1:=  \frac{1}{2}\sum_{u\in\U}|p_U(u)-q_U(u)|$. For any two sequences of random p.m.f.'s $(P_{U^{(n)}})_{n\in\mathbb{N}}$ and $(Q_{U^{(n)}})_{n\in\mathbb{N}}$ (where for every $n\in\mathbb{N}$, $U^{(n)}$ takes values in $\U^{(n)}$, which is an arbitrary set, different from $\U^n$ -- the $n$-fold cartesian product of $\U$), we write $P_{U^{(n)}}\approx Q_{U^{(n)}}$ if $\lim_{n\to\infty}\mathbb{E}[\|P_{U^{(n)}}-Q_{U^{(n)}}\|_1]\to0$. Similarly, for any two sequences of (non-random) p.m.f.'s $(p_{U^{(n)}})_{n\in\mathbb{N}}$ and $(q_{U^{(n)}})_{n\in\mathbb{N}}$, we write $p_{U^{(n)}}\approx q_{U^{(n)}}$ if $\lim_{n\to\infty}\|p_{U^{(n)}}-q_{U^{(n)}}\|_1\to0$.

\section{Problem Definition and Statements of Results}\label{sec:problem-defn}
We study randomized function computation in the two user setting; see \Figureref{problem}. Problem is specified by a pair $(p_{XY},p_{Z|XY})$, where $X,Y,Z$ take values in finite alphabets $\X,\Y,\Z$, respectively. One user (Alice) has an input sequence $X^n$, other user (Bob) has an input sequence $Y^n$, where $(X_i,Y_i)$'s are i.i.d. according to $p_{XY}$, and Bob wants to compute an output sequence $Z^n$, which should be distributed according to $p_{Z^n|X^nY^n}(z^n|x^n,y^n):= \Pi_{i=1}^np_{Z|XY}(z_i|x_i,y_i)$. We relax the correctness condition and allow a small error in function computation: Bob may output $\hat{Z}^n$ such that $\|p_{X^nY^n\hat{Z}^n}-p_{X^nY^nZ^n}\|_1\to0$ as $n\to\infty$.
We allow only one sided communication, in which Alice sends a single message to Bob, and then Bob produces an output.
Both the users have access to private randomness only.
A $(n,2^{nR})$-code $\C_n$ for computing $(p_{XY},p_{Z|XY})$ is defined as a pair of stochastic maps $(p_{M|X^n},p_{\hat{Z}^n|MY^n})$, where alphabet of $M$ is $\{1,2,\hdots,2^{nR}\}$ (here $R$ is called the rate of the code $\C_n$). We call $p_{M|X^n}$ the encoder and $p_{\hat{Z}^n|MY^n}$ the decoder.
Upon observing $X^n$, Alice maps $X^n$ to an index $M$ according to $p_{M|X^n}$ and sends $M$ to Bob. Now Bob outputs $\hat{Z}^n$ according to $p_{\hat{Z}^n|M,Y^n}$.
Note that the code $\C_n$ induces the following joint distribution: $p_{X^nY^nM\hat{Z}^n}= p_{X^nY^n}p_{M|X^n}p_{\hat{Z}^n|MY^n}$.
We study this problem in two different settings -- one with privacy, and the other without privacy. 

\subsection{Randomized Function Computation With Privacy}\label{subsec:privacy-bob}
We say that a sequence of codes $(\C_n)_{n\in\mathbb{N}}$ for computing $(p_{XY},p_{Z|XY})$ is secure if, (i) the $L_1$-distance between the induced $p_{X^nY^n\hat{Z}^n}$ and the desired $p_{X^nY^nZ^n}$ goes to zero as block-length tends to infinity, and (ii) the average amount of additional information Bob learns about $X^n$ from the message $M$ goes to zero as block-length tends to infinity. 
\begin{defn}\label{defn:as_WithPrivacy}
For a secure computation problem $(p_{XY},p_{Z|XY})$, we say that a rate $R$ is {\em achievable} if there exists a sequence of codes $(\C_n)_{n\in\mathbb{N}}$ with rate $R$, such that for every $\epsilon>0$, there exists a large enough $n$, such that
\begin{align}
\|p_{X^nY^n\hat{Z}^n}-p_{X^nY^nZ^n}\|_1 &\leq \epsilon, \label{eq:correctness_S} \\
I(M;X^n|Y^n,\hat{Z}^n) &\leq n\epsilon. \label{eq:privacy-bob}
\end{align}
We define the optimal rate $R_S$ as the infimum of all the achievable rates. 
If the set of achievable rates is empty, we say that $R_S=\infty$; and $(p_{XY},p_{Z|XY})$ is said to be computable with asymptotic security, if $R_S$ is finite.
\end{defn}
If there exists a code with $n=1$ such that \eqref{eq:correctness_S} and \eqref{eq:privacy-bob} are satisfied with $\epsilon=0$, then we say that $(p_{XY},p_{Z|XY})$ is computable with perfect security, i.e., $(p_{XY},p_{Z|XY})$ is computable with perfect security, if there exists $p(u,x,y,z)=p_{XYZ}(x,y,z)p(u|x,y,z)$ that satisfies the following Markov chains: 
\begin{align}
& U-X-Y, \label{eq:ps_correctness} \\
& Z-(U,Y)-X, \label{eq:ps_producing-output} \\
& U-(Y,Z)-X. \label{eq:ps_privacy}
\end{align}
\begin{lem}\label{lem:as-iff-ps}
$(p_{XY},p_{Z|XY})$ is computable with asymptotic security if and only if it is computable with perfect security. 
\end{lem}
\Lemmaref{as-iff-ps} is proved in \Appendixref{one-round}. Our main result is the following theorem.
\begin{thm}\label{thm:as_rate}
Suppose $(p_{XY},p_{Z|XY})$ is computable with asymptotic security. Then
\[R_S\quad  = \quad\displaystyle \min_{\substack{p_{U|XYZ}: \\ U-X-Y \\ Z-(U,Y)-X \\ U-(Y,Z)-X \\}} I(Z;U|Y),\]
where cardinality of $\U$ satisfies $|\U|\leq|\X|\cdot|\Y|\cdot|\Z|+2$.
\end{thm}
\Theoremref{as_rate} is proved in \Sectionref{proof_general-case}. If $p_{XY}$ has full support, then we obtain a more explicit expression for the optimal rate. For this we first explicitly characterize securely computable $(p_{XY},p_{Z|XY})$, and give a direct coding scheme along with a much simpler proof of its security without resorting to the achievability proof of \Theoremref{as_rate}. Details are in \Subsectionref{full-support}.

\subsection{Randomized Function Computation Without Privacy}\label{subsec:no-privacy}
\begin{defn}\label{defn:as_WithoutPrivacy}
For a randomized function computation problem $(p_{XY},p_{Z|XY})$, we say that a rate $R$ is {\em achievable} if there exists a sequence of codes $(\C_n)_{n\in\mathbb{N}}$ with rate $R$, such that for every $\epsilon>0$, there exists a large enough $n$, such that
\begin{align}
\|p_{X^nY^n\hat{Z}^n}-p_{X^nY^nZ^n}\|_1 &\leq \epsilon. \label{eq:correctness}
\end{align}
We define the optimal rate $R_{NS}$ as the infimum of all the achievable rates.
\end{defn}
\begin{thm}\label{thm:rate_no-privacy}
$R_{NS}\quad=\quad\displaystyle \min_{\substack{p_{U|XYZ}: \\ U-X-Y \\ Z-(U,Y)-X \\}} I(X,Z;U|Y),$
where cardinality of $\U$ satisfies $|\U|\leq|\X|\cdot|\Y|\cdot|\Z|+2$.
\end{thm}
\Theoremref{rate_no-privacy} can be proved along the similar lines as \Theoremref{as_rate}; A proof is in \Appendixref{proof_no-privacy}.
\begin{remark}\label{remark:randomized-Orlitsky-Roche}
{\em
\Theoremref{rate_no-privacy} can be seen as a generalization of a result of Orlitsky and Roche for deterministic function computation \cite{OrlitskyRoche}: if $p_{Z|XY}$ is a deterministic function, i.e., $p_{Z|XY}=\mathbbm{1}_{\{Z=f(X,Y)\}}$, where $f:\X\times\Y\to\Z$ is the function being computed, then the rate expression in \Theoremref{rate_no-privacy} reduces to $\displaystyle \min_{\substack{p_{U|XYZ}}} I(X;U|Y)$, where minimization is taken over all $p_{U|XYZ}$ that satisfy $U-X-Y$ and $H(Z|U,Y)=0$. This rate expression was defined to be the conditional graph entropy of an appropriately defined graph in \cite{OrlitskyRoche}. We remark that in the case of randomized function computation, the rate expression in \Theoremref{rate_no-privacy} does not seem to reduce to graph entropy. However, if we additionally require privacy only against Alice, i.e., $I(M;Y^n,\hat{Z}^n|X^n)\leq n\epsilon$ holds, then the corresponding expression for the optimal rate ($\displaystyle \min_{\substack{p_{U|XYZ}}} I(X,Z;U|Y)$, where minimization is taken over all $p_{U|XYZ}$ that satisfy $U-X-(Y,Z)$ and $Z-(U,Y)-X$) reduces to conditional graph entropy of an appropriately defined graph (see \cite{Data16} for details).
}
\end{remark}

\section{Proof of \Theoremref{as_rate}}\label{sec:proof_general-case}
We first note that the cardinality bound of $|\U|\leq|\X|\cdot|\Y|\cdot|\Z|+2$ follows from the Fenchel-Eggleston's strengthening of Carath\'eodory's theorem \cite[pg. 310]{CsiszarKorner11}. We first prove the achievability and then prove the converse. \\

\noindent{\bf Achievability:} Our achievable scheme uses the OSRB (output statistics of random binning) framework developed by Yassaee, Aref, and Gohari \cite{YassaeeArGo14}. We sketch our proof here. A detailed proof can be found in the \Appendixref{achievability}.

Fix a $p_{U|XYZ}$ that achieves the minimum in the expression for $R_S$ in \Theoremref{as_rate}. Define $p_{UXYZ}(u,x,y,z):= p_{XYZ}(x,y,z)\times p_{U|XYZ}(u|x,y,z)$. Note that $p_{UXYZ}(u,x,y,z)=p_{XY}(x,y)\times p_{U|X}(u|x)\times p_{Z|UY}(z|u,y)$ such that $U-(Y,Z)-X$ is a Markov chain.
Now consider $(U^n,X^n,Y^n,Z^n)$, where $(U_i,X_i,Y_i,Z_i)$'s are i.i.d. according to $p_{UXYZ}$. 
To make notation less cluttered, in the following we write $p(u^n,x^n,y^n,z^n)$ to mean $\Pi_{i=1}^n p_{UXYZ}(u_i,x_i,y_i,z_i)$.
We define two random mappings on $\U^n$ as follows: to each sequence $u^n\in\U^n$, assign two bins $f\in_R[2^{nR'}]$, $m\in_R[2^{nR_M}]$, independently and uniformly at random. 
Here $m$ serves as the message from Alice to Bob in the actual problem, and $f$ serves as extra randomness. Since we do not have any extra randomness in our model, later we will get rid of this by fixing an instance of it.
The random binning induces the following random p.m.f. 
\begin{align*}
P(x^n,y^n,z^n,u^n,f,m)=p(x^n,y^n,z^n,u^n)P(f|u^n)P(m|u^n).
\end{align*}
We use the Slepian-Wolf (SW) decoder $P^{\text{SW}}(\hat{u}^n|f,m,y^n)$ to produce an estimate $\hat{u}^n$ of $u^n$ from $(f,m,y^n)$.
It turns out that if $R'+R_M>H(U|Y)$, then the SW-decoder can reconstruct $U^n$ with low probability of error using $(F,M,Y^n)$ and a knowledge of the binning (see \Appendixref{achievability} for details), i.e.,
\begin{align}
P(x^n,y^n,u^n,f,m,\hat{u}^n) \approx P(x^n,y^n,u^n,f,m)\mathbbm{1}_{\{\hat{u}^n=u^n\}} \label{eq:achievability_interim1}
\end{align}
This implies
\begin{align}
P(x^n&,y^n,u^n,z^n,f,m,\hat{u}^n) \approx P(x^n,y^n,u^n,z^n,f,m)\mathbbm{1}_{\{\hat{u}^n=u^n\}} \notag \\
&= p(x^n,y^n,z^n,u^n)P(f|u^n)P(m|u^n)\mathbbm{1}_{\{\hat{u}^n=u^n\}}. \label{eq:achievability_interim2}
\end{align}
Later when we remove $F$ by conditioning on a particular instance $f$ of it, $P(x^n,y^n,z^n|f)$ may not be distributed as $p(x^n,y^n,z^n)$, which we need for our correctness condition. But if $R'<H(U|X,Y,Z)$, then $F$ becomes asymptotically independent of $(X^n,Y^n,Z^n)$ (see \Appendixref{achievability} for details), i.e.,
\begin{align}
P(x^n,y^n,z^n,f) \approx p^U(f)p(x^n,y^n,z^n),
\end{align}
where $p^U(f)$ means that $f$ takes values uniformly at random in $[2^{nR'}]$. In order to obtain a p.m.f. that corresponds to the actual coding scheme, we can expand $P(x^n,y^n,u^n,f,m,\hat{u}^n,z^n)$ in another way as follows.
\begin{align}
P(x^n,y^n&,u^n,f,m,\hat{u}^n,z^n) \approx p^U(f)p(x^n,y^n) \notag \\
&\times P(u^n,m,\hat{u}^n|x^n,y^n,f)p(z^n|\hat{u}^n,y^n), \label{eq:achievability_interim3}
\end{align}
where we define $p(z^n|\hat{u}^n,y^n):= \Pi_{i=1}^np_{Z|UY}(z_i|\hat{u}_i,y_i)$.
Note that we only need $R'<H(U|X,Y)$ (which is implied by $R'<H(U|X,Y,Z)$) to write $P(x^n,y^n,f)\approx p^U(f)p(x^n,y^n)$ in \eqref{eq:achievability_interim3}.
It follows from \eqref{eq:achievability_interim2}-\eqref{eq:achievability_interim3} that, if $R'+R_M>H(U|Y)$ and $R'<H(U|X,Y,Z)$ hold, then there exists a fixed binning (with corresponding p.m.f. $\bar{p}$) such that
\begin{align}
\bar{p}(x^n,y^n,z^n,f) &\approx p^U(f)p(x^n,y^n,z^n), \label{eq:achievability_interim45} \\
p^U(f)p(x^n,y^n)\bar{p}(u^n,m,\hat{u}^n&|x^n,y^n,f)p(z^n|\hat{u}^n,y^n) \notag \\
\approx p(x^n,y^n,z^n&,u^n)\bar{p}(f,m|u^n)\mathbbm{1}_{\{\hat{u}^n=u^n\}}. \label{eq:achievability_interim5}
\end{align}
For simplicity, define
\begin{align}
q(x^n,y^n,z^n,u^n,f,m,\hat{u}^n) &:=  p(x^n,y^n,z^n,u^n)\bar{p}(f,m|u^n) \notag \\
&\hspace{1cm}\times \mathbbm{1}_{\{\hat{u}^n=u^n\}}, \notag \\
\hat{q}(x^n,y^n,z^n,u^n,f,m,\hat{u}^n) &:=  p^U(f)p(x^n,y^n)\notag \\
\times \bar{p}(u^n&,m,\ \hat{u}^n|x^n,y^n,f)p(z^n|\hat{u}^n,y^n). \notag
\end{align}
Note that $\hat{q}$ corresponds to the actual coding scheme, except for the extra randomness $f$. It follows from \eqref{eq:achievability_interim45}, \eqref{eq:achievability_interim5}, and the above two equations (by marginalizing $(u^n,\hat{u}^n)$ away) that there exists an $f$ with $\bar{p}(f)>0$ (and $p^U(f)>0$) such that
\begin{align}
\bar{p}(x^n,y^n,z^n|f) &\approx p(x^n,y^n,z^n), \label{eq:achievability_interim7} \\
\hat{q}(x^n,y^n,z^n,m|f)&\approx q(x^n,y^n,z^n,m|f). \label{eq:achievability_interim8}
\end{align}
Now we show that $\hat{q}(x^n,y^n,z^n,m|f)$ satisfies correctness and privacy conditions.

Note that $q(x^n,y^n,z^n|f)=\bar{p}(x^n,y^n,z^n|f)$. This together with \eqref{eq:achievability_interim7} and \eqref{eq:achievability_interim8} implies correctness, i.e., $\hat{q}(x^n,y^n,z^n|f)\approx p(x^n,y^n,z^n)$.
For the privacy condition, it can be shown using the Markov chain $U-(Y,Z)-X$ that $I(M;X^n|Y^n,Z^n)|_{q(x^n,y^n,z^n,m|f)}=0$. Now \eqref{eq:achievability_interim8}, together with the fact that mutual information is a continuous function of the distribution, implies that $I(M;X^n|Y^n,Z^n)|_{\hat{q}(x^n,y^n,z^n,m|f)}\to0$ as $n\to\infty$.
Note that
\begin{align}
\hat{q}(x^n,y^n,z^n,u^n,m,\hat{u}^n|f) &= p(x^n,y^n)\bar{p}(u^n|x^n,f)\notag \\
\times \bar{p}(m|u^n)&\bar{p}^{\text{SW}}(\hat{u}^n|f,m,y^n)p(z^n|\hat{u}^n,y^n). \notag
\end{align}
Identifying $\bar{p}(m|x^n,f)=\sum_{u^n:m=m(u^n)}\bar{p}(u^n|x^n,f)$ as the encoder, and $(\bar{p}^{\text{SW}}(\hat{u}^n|f,m,y^n),p(z^n|\hat{u}^n,y^n))$ as the decoder results in a pair of encoder-decoder ensuring the security of the coding scheme. 
Now it follows that, for every $R_M>I(X,Z;U|Y)$, there exists $R'>0$ such that $R'<H(U|X,Y,Z)$ and $R_M+R'>H(U|Y)$ hold, which implies existence of a secure coding scheme by the above analysis with rate $R_M$. Using the Markov chain $U-(Y,Z)-X$ we have $R_M>I(Z;U|Y)$. \\

\noindent{\bf Converse:} Let $(p_{XY},p_{Z|XY})$ be securely computable with asymptotic security. For each $\epsilon>0$, there is a large enough $n$ and a code $\C_n$ (with rate, say, $R$) that satisfies the following properties:
{\allowdisplaybreaks
\begin{align}
& M - X^n - Y^n, \label{eq:1r-message} \\
& \hat{Z}^n - (M,Y^n) - X^n, \label{eq:1r-producing-output} \\
& I(M;X^n|Y^n,\hat{Z}^n) \leq n\epsilon, \label{eq:1r-privacy} \\
& \|p_{X^nY^n\hat{Z}^n}-p_{X^nY^nZ^n}\|_1]\leq \epsilon. \label{eq:1r-correctness}
\end{align}
We prove that the communication rate $R$ required by $\C_n$ is lower-bounded by $R_S-\delta$, where $\delta\to0$ as $\epsilon\to0$.
\begin{align}
&nR \geq H(M) \geq H(M|Y^n) \geq I(M;X^n,\hat{Z}^n|Y^n) \notag \\
&= H(X^n,\hat{Z}^n|Y^n) - H(X^n,\hat{Z}^n|M,Y^n) \notag \\
&\stackrel{\text{(a)}}{\geq} H(X^n,Z^n|Y^n) - n\epsilon_1- H(X^n,\hat{Z}^n|M,Y^n) \notag \\
&= \sum_{i=1}^n [H(X_i,Z_i|Y_i) - H(X_i,\hat{Z}_i|X^{i-1},\hat{Z}^{i-1},M,Y^n) - \epsilon_1] \notag \\
&\stackrel{\text{(b)}}{\geq} \sum_{i=1}^n [H(X_i,\hat{Z}_i|Y_i) - \epsilon_2 - H(X_i,\hat{Z}_i|X^{i-1},M,Y^n) - \epsilon_1]  \notag \\
&= \sum_{i=1}^n [I(X_i,\hat{Z}_i;U_i|Y_i) - \epsilon_1 -\epsilon_2], \notag \\
&\hspace{3cm} \text{ where } U_i=(X^{i-1},M,Y^{i-1},Y_{i+1}^n) \notag \\
&= n\Big[\sum_{i=1}^n \frac{1}{n} I(X_i,\hat{Z}_i;U_i|Y_i,T=i)\Big] - n\epsilon_1 -n\epsilon_2 \label{eq:1r-rate-interim3} \\
&= n[I(X_T,\hat{Z}_T;U_T|Y_T,T) - \epsilon_1 -\epsilon_2] \notag \\
&= n[I(X_T,\hat{Z}_T;U_T,T|Y_T) - I(X_T,\hat{Z}_T;T|Y_T) - \epsilon_1 -\epsilon_2] \notag \\
&\stackrel{\text{(c)}}{\geq} n[I(X_T,\hat{Z}_T;U_T,T|Y_T) - \epsilon_3 -\epsilon_1 - \epsilon_2]. \label{eq:1r-rate-interim4}
\end{align}
We used the following fact in (a)-(c): if $V$ and $V'$ are two random variables taking values in the same alphabet $\mathcal{V}$ such that $\|p_V-p_{V'}\|_1\leq\epsilon\leq1/4$, then it follows from \cite[Lemma 2.7]{CsiszarKorner11} that $|H(V)-H(V')|\leq \eta\log|\mathcal{V}|$, where $\eta\to0$ as $\epsilon\to0$.
Now \eqref{eq:1r-correctness} implies (a), i.e., $H(X^n,\hat{Z}^n|Y^n)\geq H(X^n,Z^n|Y^n)-n\epsilon_1$, where $\epsilon_1\to0$ as $\epsilon\to0$.
Note that \eqref{eq:1r-correctness} implies $\|p_{X_iY_i\hat{Z}_i}-p_{X_iY_iZ_i}\|_1\leq\epsilon$, for every $i\in[n]$, which implies (b), i.e., $H(X_i,\hat{Z}_i|Y_i)\geq H(X_i,Z_i|Y_i)-\epsilon_{2i}$, $i\in[n]$, where $\epsilon_{2i}\to0$ as $\epsilon\to0$; take $\epsilon_2=\max_{i\in[n]}\{\epsilon_{2i}\}$.
The random variable $T$ in \eqref{eq:1r-rate-interim3} is independent of $(M,X^n,Y^n,Z^n,\hat{Z}^n)$ and is uniformly distributed in $\{1,2,\hdots,n\}$.
Inequality (c), i.e., $I(X_T,\hat{Z}_T;T|Y_T)\leq\epsilon_3$, with $\epsilon_3\to0$ as $\epsilon\to0$, follows from $\|p_{X_iY_i\hat{Z}_i}-p_{X_iY_iZ_i}\|_1\leq\epsilon, i\in[n]$, and that $T$ is independent of $(M,X^n,Y^n,Z^n,\hat{Z}^n)$, and is shown in \Claimref{converse-small-claim} in \Appendixref{single-letterization}.
Let $U=(U_T,T)$. Similar to the above single-letterization by which we obtain \eqref{eq:1r-rate-interim4}, we can single-letterize the conditions \eqref{eq:1r-message}-\eqref{eq:1r-correctness} (proof is in \Appendixref{single-letterization}).
Let $\delta=\epsilon_1+\epsilon_2+\epsilon_3$. Now continuing from \eqref{eq:1r-rate-interim4}:
\begin{align}
R&\geq \displaystyle \min_{\substack{p_{U\hat{Z}_T|X_TY_T}: \\ U-X_T-Y_T \\ \hat{Z}_T-(U,Y_T)-X_T \\ I(U;X_T|Y_T,\hat{Z}_T)\leq\epsilon' \\ \|p_{X_TY_T\hat{Z}_T}-p_{X_TY_TZ_T}\|_1 \leq \epsilon}} I(X_T,\hat{Z}_T;U|Y_T) - \delta \label{eq:1r-rate-interim45} \\
&= \displaystyle \min_{\substack{p_{U\hat{Z}_T|XY}: \\ U-X-Y \\ \hat{Z}_T-(U,Y)-X \\ I(U;X|Y,\hat{Z}_T)\leq\epsilon' \\ \|p_{XY\hat{Z}_T}-p_{XYZ}\|_1 \leq \epsilon}} I(X,\hat{Z}_T;U|Y) - \delta \label{eq:1r-rate-interim5} \\
&\geq \displaystyle \min_{\substack{p_{UZ|XY}: \\ U-X-Y \\ Z-(U,Y)-X \\ U-(Y,Z)-X \\}} I(X,Z;U|Y) - \epsilon_4 - \delta \label{eq:1r-rate-interim6}
\end{align}
where $\delta'=\delta+\epsilon_4$ and $\delta'\to0$ as $\epsilon\to0$. In \eqref{eq:1r-rate-interim5} we used $p_{X_TY_T}=p_{XY}$ and $p_{Z_T|X_TY_T}=p_{Z|XY}$. In \eqref{eq:1r-rate-interim6} we used continuity of mutual information, continuity of $L_1$-norm, and the fact that if $(p_{XY},p_{Z|XY})$ is computable with asymptotic security, it can also be computed with perfect security (details are at the end of \Appendixref{single-letterization}). By using privacy condition $U-(Y,Z)-X$, the objective function in \eqref{eq:1r-rate-interim5} can be simplified as $I(X,Z;U|Y) = I(Z;U|Y)$. 
}
\subsection{$p_{XY}$ with full support: characterization and a simpler rate-optimal code}\label{subsec:full-support}
If $p_{XY}$ has full support, then we find an explicit $U$ that satisfies \eqref{eq:ps_correctness}-\eqref{eq:ps_privacy}, which will simplify the rate-expression in \Theoremref{as_rate} (see \Theoremref{rate_full-support}).  
For this, we first characterize $(p_{XY},p_{Z|XY})$ that can be computed with perfect security, where $p_{XY}$ has full support.
None of the proofs in this section depends on the specific distribution of $p_{XY}$ as long as it has full support. So the characterization remains the same for all $p_{XY}$ that have full support.

For every $y\in\Y$, define a set $\Z^{(y)}=\{z\in\Z:\exists x\in\X \text{ s.t. } p_{Z|XY}(z|x,y)>0\}$. Essentially, the set $\Z^{(y)}$ discards all those elements of $\Z$ that never appear as an output when Bob's input is $y$.
\begin{defn}\label{defn:equiv-relation}
For $y\in\Y$, define a relation $\equiv_y$ on the set $\Z^{(y)}$ as follows: for $z,z'\in\Z^{(y)}$, we say that $z\equiv_y z'$ if there is a sequence $z^{(1)},z^{(2)},\hdots,z^{(l-1)},z^{(l)}$ such that $z=z^{(1)}, z'=z^{(l)}$, and $z^{(i)}\sim_y z^{(i+1)}$, for every $i\in\{1,2,\hdots,l-1\}$.
\end{defn}
It is easy to see that $\equiv_y$ is an equivalence relation for every $y\in\Y$, which partitions $\Z^{(y)}$ into equivalent classes. Consider a $y\in\Y$, and let $\Z^{(y)}=\Z_{1}^{(y)}\biguplus\Z_{2}^{(y)}\biguplus\hdots\biguplus\Z_{k(y)}^{(y)}$, where $k(y)$ is the number of equivalence classes in the partition generated by $\equiv_y$. 
For every equivalence class $\Z_{i}^{(y)}$ in this partition, define a $|\X|\times|\Z_{i}^{(y)}|$ matrix $A_i^{(y)}$ such that $A_i^{(y)}(x,z)=p_{Z|XY}(z|x,y)$ for every $(x,z)\in\X\times\Z_{i}^{(y)}$. 
Note that, for every $i\in[k(y)]$, all the columns of $A_i^{(y)}$ are multiples of each other. 
Since $A_i^{(y)}$ is a rank-one matrix, we can write it as $A_i^{(y)}=\vec{\alpha}_i^{(y)}\vec{\gamma}_i^{(y)}$, where $\vec{\alpha}_i^{(y)}$ is a column vector (whose entries are non-negative and sum up to one, which makes it a unique probability vector) and $\vec{\gamma}_i^{(y)}$ is a row vector. Entries of $\vec{\alpha}_i^{(y)}$ and $\vec{\gamma}_i^{(y)}$ are indexed by $x\in\X$ and $z\in\Z_{i}^{(y)}$, respectively. 
So $A_i^{(y)}(x,z)=\vec{\alpha}_i^{(y)}(x)\vec{\gamma}_i^{(y)}(z)$.
Note that if $\vec{\alpha}_i^{(y)}=\vec{\alpha}_j^{(y)}$, then $i=j$. 

Suppose $(p_{XY},p_{Z|XY})$ is computable with perfect security, which implies that there exists $p(u,x,y,z)=p_{XYZ}(x,y,z)p(u|x,y,z)$ that satisfies \eqref{eq:ps_correctness}-\eqref{eq:ps_privacy}.
Note that the random variable $U$ corresponds to the message that Alice sends to Bob.
We define a set $\U_{i}^{(y)}$ to be the set of all those messages that Alice can send to Bob, and when Bob, having $y$ as his input outputs an element of $\Z_{i}^{(y)}$, as follows:
\begin{equation}
\U_{i}^{(y)}=\{u\in\U:p(u,z|y)>0 \text{ for some }z\in\Z_{i}^{(y)}\}. \label{eq:one-round_set-msgs}
\end{equation}
Note that for every $y\in\Y$ and $i\in[k(y)]$, the probability vector $\vec{\alpha}_i^{(y)}$ corresponds to the equivalence class $\Z_{i}^{(y)}$.
The following claim is proved in \Appendixref{characterization}.

\begin{claim}\label{claim:disjoint-messages}
Consider any $y,y'\in\Y$ and $i\in[k(y)]$, $j\in[k(y')]$. If $\vec{\alpha}_i^{(y)}\neq \vec{\alpha}_j^{(y')}$, then $\U_{i}^{(y)}\cap\U_{j}^{(y')}=\phi$.
\end{claim}
With the help of \Claimref{disjoint-messages} we have (proof is in \Appendixref{characterization}) that for every $y,y'\in\Y$, the corresponding collections of probability vectors $\{\vec{\alpha}_i^{(y)}:i\in[k(y)]\}$ and $\{\vec{\alpha}_j^{(y')}:j\in[k(y')]\}$ are equal.
\begin{claim}\label{claim:equal-alpha-vectors}
For all $y,y'\in\Y$, we have $k(y)=k(y')=:k$, and $\{\vec{\alpha}_1^{(y)},\vec{\alpha}_2^{(y)},\hdots,\vec{\alpha}_{k}^{(y)}\}=\{\vec{\alpha}_1^{(y')},\vec{\alpha}_2^{(y')},\hdots,\vec{\alpha}_{k}^{(y')}\}$.
\end{claim}
For ease of notation, without loss of generality, we rearrange the indices, to have $\vec{\alpha}_j^{(y')}=\vec{\alpha}_i^{(y)}$ if and only if $i=j$.
Now it follows from \Claimref{disjoint-messages} and \Claimref{equal-alpha-vectors} that the message set $\U$ and the alphabet $\Z^{(y)}$, for every $y\in\Y$, can be partitioned into $k$ parts as follows:
\begin{align}
\U &= \U_1\uplus\U_2\uplus\hdots\uplus\U_k, \notag \\
\Z^{(y)} &= \Z_{1}^{(y)}\uplus\Z_{2}^{(y)}\uplus\hdots\uplus\Z_{k}^{(y)}, \notag
\end{align}
where $\Z_{i}^{(y)} = \{z\in\Z^{(y)}: p(u,z|x,y)>0\text{ for some }x\in\X, u\in\U_i\}$. Note that the same $\U_i$ is used to define $\Z_{i}^{(y)}$ (and corresponds to $\vec{\alpha}_i^{(y)}$ also) for every $y\in\Y$.
Now we can state the characterization theorem, which is proved in \Appendixref{characterization}.
\begin{thm}\label{thm:one-round-characterization}
Suppose $p_{XY}$ has full support. Then $(p_{XY},p_{Z|XY})$ is computable with perfect security if and only if the following holds for every $y,y'\in\Y${\em :}
\begin{enumerate}
\item $\{\vec{\alpha}_1^{(y)},\vec{\alpha}_2^{(y)},\hdots,\vec{\alpha}_{k}^{(y)}\}=\{\vec{\alpha}_1^{(y')},\vec{\alpha}_2^{(y')},\hdots,\vec{\alpha}_{k}^{(y')}\}$.
\item For any $i\in[k]$, $\sum_{z\in\Z_{i}^{(y)}}p_{Z|XY}(z|x,y)=\sum_{z\in\Z_{i}^{(y')}}p_{Z|XY}(z|x,y')$ for every $x\in\X$.
\end{enumerate}
\end{thm}
We introduce a new random variable $W$, which takes values in $[k]$, and is jointly distributed with $(X,Y,Z)$ as follows: define $p_{XYWZ}(x,y,i,z):= 0$, if $p_{XYZ}(x,y,z)=0$; otherwise, define
\begin{align}
p_{XYWZ}&(x,y,i,z) \notag \\
&\quad:= p_{XY}(x,y)\times p_{W|X}(i|x)\times p_{Z|WY}(z|i,y) \notag \\
&\quad:=  p_{XY}(x,y)\times \Big(\sum_{z\in\Z_i^{(y)}}p_{Z|XY}(z|x,y)\Big)\notag \\
&\quad\times\Big(\mathbbm{1}_{\{z\in\Z_i^{(y)}\}}\times\frac{p_{Z|XY}(z|x,y)}{\sum_{z\in\Z_i^{(y)}}p_{Z|XY}(z|x,y)}\Big). \label{eq:W-interim2}
\end{align}
Comments are in order: (i) We defined $p_{W|X}(i|x)$ to be $\sum_{z\in\Z_i^{(y)}}p_{Z|XY}(z|x,y)$ in \eqref{eq:W-interim2} -- this is a valid definition because $\sum_{z\in\Z_i^{(y)}}p_{Z|XY}(z|x,y)$ is same for all $y$ (see \Theoremref{one-round-characterization}). (ii) We defined $p_{Z|WY}(z|i,y)$ to be $\mathbbm{1}_{\{z\in\Z_i^{(y)}\}}\times\frac{p_{Z|XY}(z|x,y)}{\sum_{z\in\Z_i^{(y)}}p_{Z|XY}(z|x,y)}$ in \eqref{eq:W-interim2} -- this is also a valid definition because the matrix $A_i^{(y)}$ defined earlier is a rank-one matrix, and therefore, $\frac{p_{Z|XY}(z|\tilde{x},y)}{\sum_{z\in\Z_i^{(y)}}p_{Z|XY}(z|\tilde{x},y)}$ is same for all $\tilde{x}$'s for which $p_{Z|XY}(z|\tilde{x},y)>0$. 
It follows from \eqref{eq:W-interim2} that $p_{XYWZ}(x,y,i,z) = \mathbbm{1}_{\{z\in\Z_i^{(y)}\}}\times p_{XYZ}(x,y,z)$. 
Note that $W$ is a deterministic function of both $U$ as well as of $(Y,Z)$. 
Now we are ready to state the main theorem of this section.

\begin{thm}\label{thm:rate_full-support}
Suppose $(p_{XY},p_{Z|XY})$ is computable with asymptotic security, where $p_{XY}$ has full support. Then
\[R_S = H(W|Y),\]
where $H(W|Y)$ is evaluated with $p_{WY}$, obtained from \eqref{eq:W-interim2}.
\end{thm}
\begin{proof}
Although the general achievability of \Theoremref{as_rate} is applicable here, we give a more direct coding scheme with a simpler argument of its correctness and privacy.
Alice passes $X^n$ through the virtual DMC $p_{W|X}$ and obtains an i.i.d. sequence $W^n$, where $p_{W|X}$ is obtained from \eqref{eq:W-interim2}. Note that if Bob gets access to this $W^n$ sequence exactly, then perfect correctness can be achieved -- Bob outputs an i.i.d. sequence $Z^n$ by passing $(W^n,Y^n)$ through the virtual DMC $p_{Z|WY}$, where $p_{Z|WY}$ is obtained from \eqref{eq:W-interim2}. It turns out that conveying $W^n$ to Bob also maintains privacy against Bob.
Now Alice can use the Slepian-Wolf coding scheme to convey $W^n$ to Bob at a rate $H(W|Y)$. We show in \Appendixref{characterization} that this scheme is asymptotically secure.
For the converse, note that $W$ is a function of both $U$ as well as of $(Y,Z)$, and using this we can simplify the objective function in \eqref{eq:1r-rate-interim6}: $I(Z;U|Y)= I(Z,W;U,W|Y)\geq H(W|Y)$.
\end{proof}

\section*{Acknowledgements}
Deepesh Data would like to thank Gowtham Raghunath Kurri for explaining to him the OSRB framework.
This work was done in part while Deepesh Data was visiting IDC Herzliya, Israel, when he was supported by ERC under the EU's Seventh Framework Programme (FP/2007-2013) ERC Grant Agreement n. 307952.
Deepesh Data's research was supported in part by a Microsoft Research India Ph.D. Fellowship.
Vinod Prabhakaran's research was supported in part by a Ramanujan Fellowship from the Department of Science and Technology, Government of India and by Information Technology Research Academy (ITRA), Government of India under ITRA Mobile grant ITRA/15(64)/Mobile/USEAADWN/01.
\bibliographystyle{IEEEtran}
\bibliography{crypto}

\appendices

\section{Preliminaries}\label{app:prelims}
\subsection{Notation}\label{subsec:notations}
We abbreviate discrete memoryless channel by DMC and independent and identically distributed by i.i.d.
We write $p^U$ to denote the uniform distribution over the corresponding alphabet.
We write capital letters $P, Q,$ etc., to denote random p.m.f.'s (more on this in \Subsectionref{YassaeeArGo14-paper}), and small letters $p,q,$ etc., to denote non-random p.m.f.'s
The total variation distance between two (non-random) p.m.f.'s $p_U$ and $q_U$ on the same alphabet $\U$ is defined by $\|p_U-q_U\|_1:= \frac{1}{2}\sum_{u\in\U}|p_U(u)-q_U(u)|$. Note that $\|p_U-q_U\|_1\leq1$.
\begin{defn}\label{defn:p.m.f.-approx}
For any two random p.m.f.'s $P_U$ and $Q_U$ on $\U$, we write $P_U\stackrel{\epsilon}{\approx}Q_U$ if $\mathbb{E}[\|P_U-Q_U\|_1]<\epsilon$. Similarly, for any two (non-random) p.m.f.'s $p_U$ and $q_U$ on $\U$, we write $p_U\stackrel{\epsilon}{\approx}q_U$ if $\|p_U-q_U\|_1<\epsilon$.
\end{defn}
\begin{defn}\label{defn:p.m.f.-sequences-approx}
For any two sequences of random p.m.f.'s $(P_{U^{(n)}})_{n\in\mathbb{N}}$ and $(Q_{U^{(n)}})_{n\in\mathbb{N}}$ (where for every $n\in\mathbb{N}$, $U^{(n)}$ takes values in $\U^{(n)}$ which is an arbitrary set, different from $\U^n$ -- the $n$-fold cartesian product of $\U$), we write $P_{U^{(n)}}\approx Q_{U^{(n)}}$ if $\lim_{n\to\infty}\mathbb{E}[\|P_{U^{(n)}}-Q_{U^{(n)}}\|_1]\to0$. Similarly, for any two sequences of (non-random) p.m.f.'s $(p_{U^{(n)}})_{n\in\mathbb{N}}$ and $(q_{U^{(n)}})_{n\in\mathbb{N}}$, we write $p_{U^{(n)}}\approx q_{U^{(n)}}$ if $\lim_{n\to\infty}\|p_{U^{(n)}}-q_{U^{(n)}}\|_1\to0$.
\end{defn}

\subsection{Output Statistics of Random Binning}\label{subsec:YassaeeArGo14-paper}
Some of our achievability proofs use the OSRB (output statistics of random binning) framework developed by Yassaee, Aref, and Gohari \cite{YassaeeArGo14} and their results. 
We simplify the setting and statements of these results for our purpose.
Let $(U,V)$ be a discrete memoryless correlated source distributed according to a p.m.f. $p_{UV}$, where $U$ and $V$ take values in finite alphabets $\U$ and $\V$, respectively. A random binning $\B$ consists of a set of two random mappings $\B_i:\U^n\to[2^{nR_i}]$, $i\in\{1,2\}$, in which $\B_i$ maps each sequence $u^n\in\U^n$ uniformly and independently to an element in $[2^{nR_i}]$. In other words, the random binning $\B_i$ is a random partition of the set $\U^n$ into $2^{nR_i}$ bins. In the following we write $p(u^n,v^n)$ to denote $\Pi_{i=1}^np_{UV}(u_i,v_i)$. A random binning induces the following random p.m.f. on the set $\U^n\times\V^n\times[2^{nR_1}]\times[2^{nR_2}]$:
\begin{align*}
P(u^n,v^n,b_1,b_2) = p(u^n,v^n)\times\Pi_{i=1}^2\mathbbm{1}_{\{\B_i(u^n)=b_i\}},
\end{align*}
where capital $P$ is used to indicate the p.m.f. induced on $(u^n,v^n,b_1,b_2)$ is random. For $i\in\{1,2\}$, the following theorem finds constraints on the rate $R_i$, such that $V^n$ becomes asymptotically independent of the bin $\B_i(U^n)$ in expectation, where expectation is taken over the random binning $\B_i$. 
\begin{thm}{\cite[Theorem 1]{YassaeeArGo14}}\label{thm:source-bin-indep}
For $i\in\{1,2\}$, if $R_i<H(U|V)$, then as $n$ goes to infinity, we have 
\begin{align*}
\mathbb{E}_{\B_i}[\|P(v^n,b_i)-p(v^n)\times p^U(b_i)\|_1]\to0,
\end{align*}
where $\B_i$ is the set of all random mappings $\B_i:\U^n\to[2^{nR_i}]$.
\end{thm}
Now we write the achievability statement of Slepian-Wolf theorem in a different equivalent from. Note that achievability proof of the Slepian-Wolf theorem \cite[Section 10.3.2]{ElgamalKim11} gives (i) a constraint on the rate $R_1+R_2$ that allows reconstruction of $U^n$ from the bin indices $\B_1(U^n),\B_2(U^n)$ and $V^n$, and (ii) a decoder (with respect to any fixed binning), which outputs $\hat{u}^n$ in the presence of bin $(b_1,b_2)$ and the sequence $v^n$, such that the probability of reconstruction error goes to zero as $n$ tends to infinity.
Here we denote the decoder by the random conditional p.m.f. $P^{\text{SW}}(\hat{u}^n|v^n,b_1,b_2)$. 
Since the decoder is a function, this random p.m.f. $P^{\text{SW}}(\hat{u}^n|v^n,b_1,b_2)$ takes only two values 0 and 1: for a Slepian-Wolf decoder that uses a jointly typical decoder, $P^{\text{SW}}(\hat{u}^n|v^n,b_1,b_2)=1$ if $\hat{u}^n$ is the unique jointly typical sequence with $v^n$ in the bin $(b_1,b_2)$. If there does not exists the unique sequence $\hat{u}^n$ in the bin $(b_1,b_2)$ that is jointly typical with $v^n$, then $\hat{u}^n$ is taken to be a fixed arbitrary sequence.
\begin{lem}{\cite[Lemma 1]{YassaeeArGo14}}\label{lem:slepian-wolf-condition}
If $R_1+R_2>H(U|V)$, then as $n$ goes to infinity, we have
\begin{align*}
\mathbb{E}_{\B}[\|P(u^n,v^n,\hat{u}^n)-p(u^n,v^n)\times\mathbbm{1}_{\{\hat{u}^n=u^n\}}\|_1]\to0,
\end{align*}
where $\B$ is the set of all random mappings $\B_i:\U^n\to[2^{nR_i}], i\in\{1,2\}$.
\end{lem}

\section{Details omitted from \Sectionref{problem-defn}}\label{app:one-round}
\begin{proof}[Proof of \Lemmaref{as-iff-ps}]
This can be proved along the lines of the proof of Theorem 3 in \cite{Data16}.
\end{proof}

\section{Achievability proof of \Theoremref{as_rate}}\label{app:achievability}
\noindent{\bf Achievability:} Our achievable scheme uses the OSRB (output statistics of random binning) framework developed by Yassaee, Aref, and Gohari \cite{YassaeeArGo14}. According to the OSRB framework, we divide our achievability proof in three parts: in part (1) we define two protocols, protocol A and protocol B. Each of these protocols will induce a p.m.f. on random variables defined during the protocol (see \Figureref{combined}). Protocol A corresponds to the source coding side of the problem (see \Figureref{source-coding}) and does not lead to a coding algorithm; protocol B almost gives a coding scheme, except for the fact that it is assisted with shared randomness (see \Figureref{protocol})), which does not exist in the problem that we are trying to solve. In part (2) our goal is to find conditions that will make these two distributions almost identical. Once we find such conditions (which will make protocol A and protocol B almost identical) we can investigate the correctness and privacy properties, that we want protocol B to satisfy, in protocol A. Notice that protocol B gives a coding scheme only if there is no shared randomness. We get rid of this in part (3) by finding an instance of a shared randomness in protocol B, conditioned on which both correctness and privacy conditions are still preserved. Details are given below.

\begin{figure*}[h]
\begin{subfigure}{0.5\textwidth}
\includegraphics[scale=0.75]{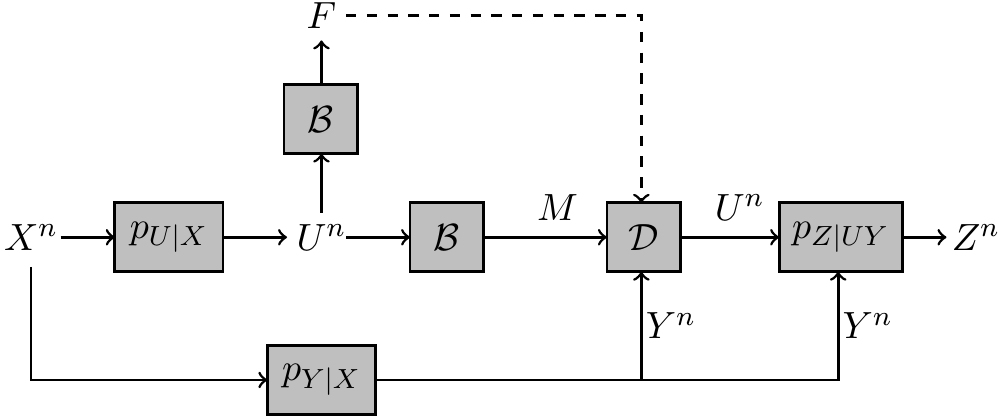}
\caption{Source coding side of the problem}
\label{fig:source-coding}
\end{subfigure}
\begin{subfigure}{0.5\textwidth}\centering
\includegraphics[scale=0.75]{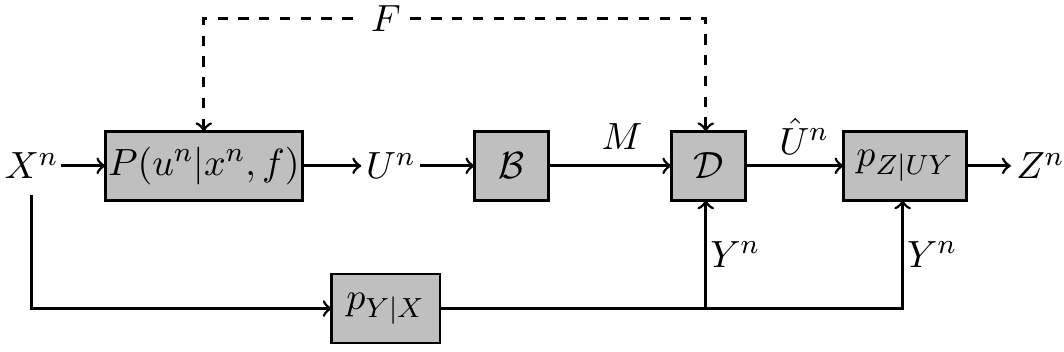}
\caption{Coding for the actual problem assisted with the extra shared randomness}
\label{fig:protocol}
\end{subfigure}
\caption{
(\Figureref{source-coding}.) The source coding side of the problem (Protocol A). Given an i.i.d. source $(X^n,Y^n)$, we pass the $X^n$ sequence through a virtual DMC $p_{U|X}$ to get an i.i.d. sequence $U^n$, and then we pass $(U^n,Y^n)$ sequence through another virtual DMC $p_{Z|UY}$ to get an i.i.d. sequence $Z^n$. Here, $(X^n,Y^n,Z^n)$ have the desired i.i.d. distribution according to $p_{X^nY^nZ^n}(x^n,y^n,z^n)=\Pi_{i=1}^np_{XYZ}(x_i,y_i,z_i)$. Since the decoder needs access to the $U^n$ sequence to produce $Z^n$, we have to describe the $U^n$ sequence to decoder. We do this by assigning two random and independent bins $M$ and $F$ to $U^n$ sequences, where $M$ will serve as the message from Alice to Bob in the actual problem, and $F$ will serve as extra randomness. Bob uses the Slepian-Wolf decoder for estimating $U^n$ from $(M,F,Y^n)$. If the SW constraint in \eqref{eq:part2-constraint3} is satisfied, Bob will be able to estimate the $U^n$ sequence using the SW decoder with low probability of error. Note that the $U^n$ sequence that we are feeding into the DMC $p_{Z|UY}$ is not equal to the output of the SW decoder, but it is the i.i.d. $U^n$ sequence that we obtained earlier by passing $X^n$ through the DMC $p_{U|X}$. (\Figureref{protocol}.) Coding for the actual problem assisted with extra shared randomness (Protocol B). Given an i.i.d. source $(X^n,Y^n)$, we pass $X^n$ and the extra shared randomness $F$ through the reverse encoder $P(u^n|x^n,f)$ (which is obtained from protocol A) to get a $U^n$ sequence ($U_i$'s may not be i.i.d.). Alice uses the p.m.f. obtained from protocol A to map this $U^n$ sequence to a message $M$ and sends it to Bob. Bob uses the SW decoder of protocol A to decode $\hat{U}^n$ from $(M,F,Y^n)$, and produces a $Z^n$ sequence (where $Z_i$'s may not be i.i.d.) by passing $(\hat{U}^n,Y^n)$ through $p_{Z|UY}$. Note that in protocol B, $F$ and $(X^n,Y^n)$ are independent, whereas in protocol A, they may not be. To make the two p.m.f.'s (induced by protocol A and protocol B) almost identical, we need to have the asymptotic independence between $F$ and $(X^n,Y^n)$ in protocol A, which is ensured by \eqref{eq:part2-constraint1}. Finally, since we do not have any shared randomness in the problem that we trying to solve, we must eliminate $F$ by conditioning on an instance of shared randomness without disturbing the desired joint distribution of $(X^n,Y^n,Z^n)$. This is ensured by \eqref{eq:part2-constraint2}, which makes $F$ and $(X^n,Y^n,Z^n)$ almost independent in protocol A.
}
\label{fig:combined}
\end{figure*}

\noindent {\it Part (1) of the proof:} We will define two protocols here, one of which is related to the source coding side of the problem (see \Figureref{source-coding}) and the other one is related to the actual protocol (with shared randomness) (see \Figureref{protocol}).

\textit{Protocol A:} Fix a $p_{U|XYZ}$ that achieves the minimum in the expression for $R_S$ in \Theoremref{as_rate}. Define $p_{UXYZ}(u,x,y,z):=p_{XYZ}(x,y,z)\times p_{U|XYZ}(u|x,y,z)$. Note that $p_{UXYZ}(u,x,y,z)=p_{XY}(x,y)\times p_{U|X}(u|x)\times p_{Z|UY}(z|u,y)$ such that $U-(Y,Z)-X$ is a Markov chain.
Now consider $(U^n,X^n,Y^n,Z^n)$, where $(U_i,X_i,Y_i,Z_i)$'s are i.i.d. according to $p_{UXYZ}$. To make the notation less cluttered, in the following we write $p(u^n,x^n,y^n,z^n)$ to mean $\Pi_{i=1}^n p_{UXYZ}(u_i,x_i,y_i,z_i)$.
\begin{enumerate}
\item Map each sequence $u^n\in\U^n$ into two bins $f\in_R[2^{nR'}]$ and $m\in_R[2^{nR_M}]$, independently and uniformly at random.
\item We will use Slepian-Wolf decoder to estimate $u^n$ from $(f,m,y^n)$.
\end{enumerate}
In part (2) and part (3) of the proof, we will impose constraints on the rates $R'$ and $R_M$ that will imply that $R_M\geq I(Z;U|Y)$. Now we will see the joint distribution induced by the protocol A as depicted in \Figureref{source-coding}. We denote the random p.m.f. (p.m.f. is random because binning is random) induced by the protocol A by $P$ as follows:
{\allowdisplaybreaks
\begin{align}
P(x^n,y^n,&u^n,z^n,f,m,\hat{u}^n)  \notag \\
&= p(x^n,y^n,u^n)\times p(z^n|u^n,y^n)\times P(f|u^n)\times P(m|u^n)\times P^{\text{SW}}(\hat{u}^n|f,m,y^n) \label{eq:protocol-A-interim1} \\
&= p(x^n,y^n,u^n)\times P(f|u^n)\times P(m|u^n)\times P^{\text{SW}}(\hat{u}^n|f,m,y^n)\times p(z^n|u^n,y^n) \notag \\
&= P(x^n,y^n,u^n,f)\times P(m|u^n)\times P^{\text{SW}}(\hat{u}^n|f,m,y^n)\times p(z^n|u^n,y^n) \notag \\
&= P(x^n,y^n,f)\times P(u^n|x^n,y^n,f)\times P(m|u^n)\times P^{\text{SW}}(\hat{u}^n|f,m,y^n)\times p(z^n|u^n,y^n) \notag \\
&= P(x^n,y^n,f)\times P(u^n|x^n,f)\times P(m|u^n)\times P^{\text{SW}}(\hat{u}^n|f,m,y^n)\times p(z^n|u^n,y^n) \label{eq:protocol-A}
\end{align}
In \eqref{eq:protocol-A-interim1} we used the fact that $p(u,x,y,z)$ satisfies $Z-(Y,U)-X$. In \eqref{eq:protocol-A} we used the Markov chain $U^n-(X^n,F)-Y^n$ (which is shown below) to write $P(u^n|x^n,y^n,f)=P(u^n|x^n,f)$.
\begin{align*}
I(U^n;Y^n|X^n,F) &\leq I(U^n,F;Y^n|X^n) \\
&= I(U^n;Y^n|X^n) + I(F;Y^n|U^n,X^n) \\
&\leq I(U^n;Y^n|X^n) + I(F;X^n,Y^n|U^n) \\ 
&= 0
\end{align*}
where in the last equality we used $I(U^n;Y^n|X^n)=0$ (which follows because $I(U^n;Y^n|X^n)=n I(U;Y|X)$, and that  $U-X-Y$ is a Markov chain) and $I(F;X^n,Y^n|U^n)=0$ (which follows from the fact that conditioned on $U^n$ the bin index $F$ is independent of $(X^n,Y^n)$). \\
}

\textit{Protocol B:} In this protocol we assume that Alice and Bob have access to shared randomness, denoted by $F$, which is uniformly distributed in $[2^{nR'}]$. Note that $F$ is independent of $(X^n,Y^n)$.
\begin{enumerate}
\item Alice samples a sequence $u^n$ according to $P(u^n|x^n,f)$ and then the bin index $m$ according to $P(m|u^n)$ (where these distributions are defined in protocol A) and sends $m$ to Bob.
\item Bob, having access to $f,m,y^n$ and a knowledge of the binning, uses the Slepian-Wolf decoder $P^{\text{SW}}(\hat{u}^n|f,m,y^n)$ (from protocol A) to obtain $\hat{u}^n$ -- an estimate of $u^n$ -- and outputs $z^n$ according to $p(z^n|\hat{u}^n,y^n)$ (which is equal to $\Pi_{i=1}^np_{Z|UY}(z_i|\hat{u}_i,y_i)$).
\end{enumerate}
Let $\hat{P}(x^n,y^n,u^n,z^n,f,m,\hat{u}^n)$ denote the random p.m.f. induced by the protocol B, which is described as follows:
\begin{align}
\hat{P}(x^n,y^n,u^n,z^n,f,m,\hat{u}^n) &= p(x^n,y^n)\times p^{U}(f)\times P(u^n|x^n,f)\times P(m|u^n) \notag \\ 
&\hspace{3cm} \times P^{\text{SW}}(\hat{u}^n|f,m,y^n)\times p(z^n|\hat{u}^n,y^n) \label{eq:protocol-B}
\end{align}

\noindent \textit{Part (2) of the proof:} Here we find some constraints on the rates $R_M,R'$ that will make the p.m.f.'s $P$ (from protocol A) and $\hat{P}$ (from protocol B) close in total variation distance. It suffices to find constraints that will make $P(x^n,y^n,f)$ close to $p(x^n,y^n)\times p^U(f)$ and the Slepian-Wolf decoder to reliably decode the sequence $U^n$.
\begin{enumerate}
\item Note that $f$ is a bin index of the sequence $u^n$. So, it follows from \Theoremref{source-bin-indep} (by putting $V=(X,Y)$) that if
\begin{align}
R'<H(U|X,Y),\label{eq:part2-constraint1}
\end{align} 
then $P(x^n,y^n,f)\approx p^U(f)\times p(x^n,y^n)=\hat{P}(x^n,y^n,f)$. This implies (by marginalizing $z^n$ away from \eqref{eq:protocol-A} and \eqref{eq:protocol-B}) that
\begin{align}
P(x^n,y^n,u^n,f,m,\hat{u}^n) \approx \hat{P}(x^n,y^n,u^n,f,m,\hat{u}^n). \label{eq:part2-interim1}
\end{align}
Note that asymptotic independence of $F$ and $(X^n,Y^n)$ is enough to establish \eqref{eq:part2-interim1}. However, we will need asymptotic independence of $F$ and $(X^n,Y^n,Z^n)$ in $P$ (in part (3) of the proof) to ensure correctness of our protocol when we remove the shared randomness between Alice and Bob in protocol B (after fixing a binning in protocol A). It follows from \Theoremref{source-bin-indep} (by putting $V=(X,Y,Z)$) that if
\begin{align}
R' < H(U|X,Y,Z), \label{eq:part2-constraint2}
\end{align}
then 
\begin{align}
P(x^n,y^n,z^n,f)\approx p^U(f)\times p(x^n,y^n,z^n). \label{eq:part2-p.m.f.-approx1}
\end{align}
Note that \eqref{eq:part2-constraint2} implies \eqref{eq:part2-constraint1}.
\vspace{0.25cm}
\item It follows from \Lemmaref{slepian-wolf-condition} (by putting $V=Y$) that if 
\begin{align}
R'+R_M>H(U|Y) \label{eq:part2-constraint3}
\end{align}
then the Slepian-Wolf decoder of protocol A will decode the sequence $U^n$ with low probability of error, i.e., 
\begin{align}
P(x^n,y^n,u^n,f,m,\hat{u}^n) \approx P(x^n,y^n,u^n,f,m)\times \mathbbm{1}_{\{\hat{u}^n=u^n\}}. \label{eq:part2-interim2}
\end{align}
\end{enumerate}
\eqref{eq:part2-interim1} and \eqref{eq:part2-interim2} imply
\begin{align}
\hat{P}(x^n,y^n,u^n,f,m,\hat{u}^n) \approx P(x^n,y^n,u^n,f,m)\times \mathbbm{1}_{\{\hat{u}^n=u^n\}}. \label{eq:part2-interim3}
\end{align}
Now, using \eqref{eq:part2-interim2} and \eqref{eq:part2-interim3} we can show that the p.m.f.'s $P$ and $\hat{P}$ are close in total variation distance as follows:
{\allowdisplaybreaks
\begin{align}
\hat{P}(x^n,y^n,u^n,z^n,f,m,\hat{u}^n) &= \hat{P}(x^n,y^n,u^n,f,m,\hat{u}^n)\times p(z^n|\hat{u}^n,y^n) \notag \\
&\approx P(x^n,y^n,u^n,f,m)\times \mathbbm{1}_{\{\hat{u}^n=u^n\}}\times p(z^n|\hat{u}^n,y^n) \label{eq:part2-interim4} \\
&= P(x^n,y^n,u^n,f,m)\times \mathbbm{1}_{\{\hat{u}^n=u^n\}}\times p(z^n|u^n,y^n) \notag \\
&\approx P(x^n,y^n,u^n,f,m,\hat{u}^n)\times p(z^n|u^n,y^n) \label{eq:part2-interim5} \\
&= P(x^n,y^n,u^n,z^n,f,m,\hat{u}^n) \label{eq:part2-interim6}
\end{align}
\eqref{eq:part2-interim4} and \eqref{eq:part2-interim5} follow from \eqref{eq:part2-interim3} and \eqref{eq:part2-interim2}, respectively. Marginalizing $u^n,\hat{u}^n$ away from \eqref{eq:part2-interim6} gives
\begin{align}
\hat{P}(x^n,y^n,z^n,f,m) \approx P(x^n,y^n,z^n,f,m). \label{eq:part2-p.m.f.-approx2}
\end{align}
}

\noindent \textit{Part (3) of the proof:} Note that there is no shared randomness between Alice and Bob in the problem that we are trying to solve. Now we find an instance $f$ of the shared randomness, conditioned on which $\hat{P}(x^n,y^n,u^n,z^n,m|f)$ satisfies correctness and privacy.
It follows from \eqref{eq:part2-p.m.f.-approx1} and \eqref{eq:part2-p.m.f.-approx2} that \[\lim_{n\to\infty}\mathbb{E}_{\B}[\|P(x^n,y^n,z^n,f)-p^U(f)\times p(x^n,y^n,z^n)\|_1+\|\hat{P}(x^n,y^n,z^n,f,m)-P(x^n,y^n,z^n,f,m)\|_1]=0,\] where expectation is taken over random binning.
This implies that there exists a fixed binning (with the corresponding p.m.f. $\bar{p}$) such that if we replace $P$ with $\bar{p}$ in \eqref{eq:protocol-B} and denote the resulting p.m.f. with $\hat{p}$, we have
\begin{align}
\bar{p}(x^n,y^n,z^n,f) &\approx p^U(f)\times p(x^n,y^n,z^n), \label{eq:part3-interim0} \\
\hat{p}(x^n,y^n,z^n,f,m) &\approx \bar{p}(x^n,y^n,z^n,f,m). \label{eq:part3-interim1}
\end{align}
\eqref{eq:part3-interim0} and \eqref{eq:part3-interim1} imply (proved in \Claimref{existence-of-f} at the end of part (3)) that there exists an $f$ such that $\bar{p}(f)>0$ and
\begin{align}
\bar{p}(x^n,y^n,z^n|f) &\approx p(x^n,y^n,z^n), \label{eq:part3-interim15} \\
\hat{p}(x^n,y^n,z^n,m|f) &\approx \bar{p}(x^n,y^n,z^n,m|f). \label{eq:part3-interim2}
\end{align}
Note that $\hat{p}(f)=p^U(f)>0$ for every $f$, which implies that, for all $f$, if $\bar{p}(f)>0$, then $\hat{p}(f)>0$. So, all the terms in \eqref{eq:part3-interim15} and \eqref{eq:part3-interim2} are well-defined.
Now we show that $\hat{p}(x^n,y^n,z^n,m|f)$ from \eqref{eq:part3-interim2} satisfies correctness and privacy conditions. \\

Correctness: Marginalizing $m$ away from \eqref{eq:part3-interim2} and using \eqref{eq:part3-interim15} imply the correctness condition, i.e., $\hat{p}(x^n,y^n,z^n|f)\approx p(x^n,y^n,z^n)$. \\

Privacy: In order to show privacy, i.e., $I(M;X^n|Y^n,Z^n)|_{\hat{p}(x^n,y^n,z^n,m|f)}\to0$, as $n\to\infty$, we first show $\bar{p}(x^n,y^n,z^n,f,m)=p(x^n,y^n,z^n)\times \bar{p}(f,m|y^n,z^n)$:
{\allowdisplaybreaks
\begin{align}
\bar{p}(x^n,y^n,z^n,f,m) &= \sum_{u^n} \bar{p}(x^n,y^n,z^n,u^n,f,m) \notag \\
&= \sum_{u^n} p(x^n,y^n,z^n)\times p(u^n|x^n,y^n,z^n)\times \bar{p}(f,m|u^n) \notag \\
&\stackrel{\text{(a)}}{=} \sum_{u^n} p(x^n,y^n,z^n)\times p(u^n|y^n,z^n)\times \bar{p}(f,m|u^n) \notag \\
&= \sum_{u^n} p(x^n,y^n,z^n)\times \bar{p}(u^n,f,m|y^n,z^n) \notag \\
&= p(x^n,y^n,z^n)\times \bar{p}(f,m|y^n,z^n) \notag
\end{align}
where (a) follows from the fact that $p(u^n,x^n,y^n,z^n)=\Pi_{i=1}^np(u_i,x_i,y_i,z_i)$ and that $p(u,x,y,z)$ satisfies the Markov chain $U-(Y,Z)-X$. So we have
\begin{align*}
& I(F,M;X^n|Y^n,Z^n)|_{\bar{p}(x^n,y^n,z^n,f,m)}=0 \\
\implies & I(M;X^n|Y^n,Z^n,F)|_{\bar{p}(x^n,y^n,z^n,f,m)}=0 \\
\implies & I(M;X^n|Y^n,Z^n,F=f)|_{\bar{p}(x^n,y^n,z^n,f,m)}=0, \quad\forall f, \text{ s.t. } \bar{p}(f)>0 \\
\implies & I(M;X^n|Y^n,Z^n)|_{\bar{p}(x^n,y^n,z^n,m|f)}=0, \quad\forall f, \text{ s.t. } \bar{p}(f)>0 \\
\end{align*}
In particular, $I(M;X^n|Y^n,Z^n)|_{\bar{p}(x^n,y^n,z^n,m|f)}=0$ with the specific $f$ in \eqref{eq:part3-interim2}.
}
\vspace{0.25cm}
We have $\hat{p}(x^n,y^n,z^n,m|f) \approx \bar{p}(x^n,y^n,z^n,m|f)$ and $I(M;X^n|Y^n,Z^n)|_{\bar{p}(x^n,y^n,z^n,m|f)}=0$. Now, since mutual information is a continuous function of the distribution, we have 
\[I(M;X^n|Y^n,Z^n)|_{\hat{p}(x^n,y^n,z^n,m|f)}\to0, \text{ as } n\to\infty.\]

It follows from \eqref{eq:part2-constraint2} and \eqref{eq:part2-constraint3} that, for every $R_M > I(X,Z;U|Y)$ there exists $R'>0$ such that $R'<H(U|X,Y,Z)$ and $R_M+R'>H(U|Y)$ hold, which implies existence of a secure coding scheme by the above analysis with rate $R_M$. Since $U-(Y,Z)-X$ is a Markov chain, we have $R_M > I(Z;U|Y)$.
We have to give a pair of encoder and decoder that results in a distribution that satisfies correctness and privacy conditions. The encoder is specified by $\bar{p}(m|x^n,f)=\sum_{u^n:m=m(u^n)}\bar{p}(u^n|x^n,f)$, and the decoder is specified by $(\bar{p}^{\text{SW}}(\hat{u}^n|f,m,y^n),p(z^n|\hat{u}^n,y^n))$.

\begin{claim}\label{claim:existence-of-f}
Let $\epsilon_n,\delta_n\to0$ be such that
\begin{align}
\bar{p}(x^n,y^n,z^n,f) &\stackrel{\epsilon_n}{\approx} p^U(f)\times p(x^n,y^n,z^n), \label{eq:part3-interim11} \\
\hat{p}(x^n,y^n,z^n,f,m) &\stackrel{\delta_n}{\approx} \bar{p}(x^n,y^n,z^n,f,m) \label{eq:part3-interim12}
\end{align}
hold. Then there exists an $f$ such that $\bar{p}(f)>0$ and 
\begin{align*}
\bar{p}(x^n,y^n,z^n|f) &\stackrel{3\epsilon_n+\delta_n}{\approx} p(x^n,y^n,z^n), \\
\hat{p}(x^n,y^n,z^n,m|f) &\stackrel{3\epsilon_n+\delta_n}{\approx} \bar{p}(x^n,y^n,z^n,m|f).
\end{align*}
\end{claim}
\begin{proof}
Note that \eqref{eq:part3-interim11} implies $\bar{p}(f)\stackrel{\epsilon_n}{\approx} p^U(f)$ and \eqref{eq:protocol-B} implies $\hat{p}(f)=p^U(f)$. Putting these back in \eqref{eq:part3-interim11} and \eqref{eq:part3-interim12} gives:
\begin{align}
\bar{p}(f)\times \bar{p}(x^n,y^n,z^n|f) &\stackrel{2\epsilon_n}{\approx} \bar{p}(f)\times p(x^n,y^n,z^n), \label{eq:part3-interim121} \\
\bar{p}(f)\times\hat{p}(x^n,y^n,z^n,m|f) &\stackrel{\epsilon_n+\delta_n}{\approx} \bar{p}(f)\times \bar{p}(x^n,y^n,z^n,m|f). \label{eq:part3-interim122}
\end{align}
\eqref{eq:part3-interim121} and \eqref{eq:part3-interim122} imply
\begin{align}
\sum_f \bar{p}(f)\times \|\bar{p}(x^n,y^n,z^n|f)-p(x^n,y^n,z^n)\|_1 &\leq 2\epsilon_n, \label{eq:part3-interim123} \\
\sum_f \bar{p}(f)\times \|\hat{p}(x^n,y^n,z^n,m|f)-\bar{p}(x^n,y^n,z^n,m|f)\|_1 &\leq \epsilon_n+\delta_n. \label{eq:part3-interim124}
\end{align}
\eqref{eq:part3-interim123} and \eqref{eq:part3-interim124} imply
\begin{align}
\sum_f \bar{p}(f)\times \Big[\|\bar{p}(x^n,y^n,z^n|f)-p(x^n,y^n,z^n)\|_1 +  &\|\hat{p}(x^n,y^n,z^n,m|f)-\bar{p}(x^n,y^n,z^n,m|f)\|_1\Big] \notag \\
& \leq 3\epsilon_n+\delta_n. \label{eq:part3-interim125}
\end{align}
Now \eqref{eq:part3-interim125} implies that there exists an $f$ such that $\bar{p}(f)>0$ and 
\begin{align}
[\|\bar{p}(x^n,y^n,z^n|f)-p(x^n,y^n,z^n)\|_1 +  &\|\hat{p}(x^n,y^n,z^n,m|f)-\bar{p}(x^n,y^n,z^n,m|f)\|_1 \leq 3\epsilon_n+\delta_n.
\end{align}
This means that each term on the left hand side is upper-bounded by $3\epsilon_n+\delta_n$, which implies our claim.
\end{proof}

\section{Details omitted from the converse of \Theoremref{as_rate} in \Sectionref{proof_general-case}}\label{app:single-letterization}
First we show that $I(X_T,\hat{Z}_T;T|Y_T)\leq\epsilon_3$, where $\epsilon_3\to0$ as $\epsilon\to0$, which implies the inequality in \eqref{eq:1r-rate-interim4}.
\begin{claim}\label{claim:converse-small-claim}
$I(X_T,\hat{Z}_T;T|Y_T)\leq\epsilon_3$, where $\epsilon_3\to0$ as $\epsilon\to0$.
\end{claim}
\begin{proof}
{\allowdisplaybreaks
\begin{align*}
I(X_T,\hat{Z}_T;T|Y_T) &= H(X_T,\hat{Z}_T|Y_T) - H(X_T,\hat{Z}_T|Y_T,T) \\
&\stackrel{\text{(d)}}{\leq} H(X_T,Z_T|Y_T) + \delta_1 - H(X_T,\hat{Z}_T|Y_T,T) \\
&= H(X_T,Z_T|Y_T) + \delta_1 - \sum_{i=1}^n\frac{1}{n}H(X_i,\hat{Z}_i|Y_i,T=i) \\
&\stackrel{\text{(e)}}{=} H(X_T,Z_T|Y_T) + \delta_1 - \sum_{i=1}^n\frac{1}{n}H(X_i,\hat{Z}_i|Y_i) \\
&\stackrel{\text{(f)}}{\leq} H(X_T,Z_T|Y_T) + \delta_1 - \sum_{i=1}^n\frac{1}{n}[H(X_i,Z_i|Y_i)-\delta_2] \\
&\stackrel{\text{(g)}}{=} H(X_T,Z_T|Y_T) + \delta_1 + \delta_2 - \sum_{i=1}^n\frac{1}{n}H(X_i,Z_i|Y_i,T=i) \\
&= H(X_T,Z_T|Y_T) + \delta_1 + \delta_2 - H(X_T,Z_T|Y_T,T) \\
&\stackrel{\text{(h)}}{=} H(X,Z|Y) + \delta_1 + \delta_2 - H(X,Z|Y) \\
&= \epsilon_3, \quad\text{ where } \epsilon_3=\delta_1+\delta_2.
\end{align*}
We used the following fact in (d) and (f): if $V$ and $V'$ are two random variables taking values in the same alphabet $\mathcal{V}$ such that $\|p_V-p_{V'}\|_1\leq\epsilon\leq1/4$, then it follows from \cite[Lemma 2.7]{CsiszarKorner11} that $|H(V)-H(V')|\leq \eta\log|\mathcal{V}|$, where $\eta\to0$ as $\epsilon\to0$.
Now $\|p_{X_TY_T\hat{Z}_T}-p_{X_TY_TZ_T}\|_1\leq\epsilon$ (see \eqref{eq:single_correctness}) implies (d), i.e., $H(X_T,\hat{Z}_T|Y_T)\leq H(X_T,Z_T|Y_T)+\delta_1$, where $\delta_1\to0$ as $\epsilon\to0$. Similarly, $\|p_{X_iY_i\hat{Z}_i}-p_{X_iY_iZ_i}\|_1\leq\epsilon, i\in[n]$ (see \eqref{eq:single_correctness-i}) implies (f), i.e., $H(X_i,\hat{Z}_i|Y_i)\geq H(X_i,Z_i|Y_i)-\delta_{2i}$, where $\delta_{2i}\to0$ as $\epsilon\to0$; take $\delta_2=\max_{i\in[n]}\{\delta_{2i}\}$. In (e) and (g) we used the fact that $T$ is independent of $(M,X^n,Y^n,Z^n,\hat{Z}^n)$. Since $\epsilon_3=\delta_1+\delta_2$, we have $\epsilon_3\to0$ as $\epsilon\to0$. In (h) we used the fact that $p_{X_TY_TZ_T}=p_{XYZ}$ and that $T$ is independent of $(X,Y,Z)$.
}
\end{proof}
\vspace{0.25cm}
Now we single-letterize conditions \eqref{eq:1r-message}-\eqref{eq:1r-correctness}. While doing this we will introduce a time-sharing random variable $T$ that is independent of $(M,X^n,Y^n,Z^n,\hat{Z}^n)$ and uniformly distributed in $\{1,2,\hdots,n\}$. \\

{\it Single-letterizing \eqref{eq:1r-message}:}
{\allowdisplaybreaks
\begin{align*}
0 &= I(M;Y^n|X^n) \\
&= H(Y^n|X^n) - H(Y^n|X^n,M) \\
&= \sum_{i=1}^n [H(Y_i|X_i) - H(Y_i|X^n,M,Y^{i-1})] \\
&\geq \sum_{i=1}^n [H(Y_i|X_i) - H(Y_i|X_i,X^{i-1},M,Y^{i-1})] \\
&\stackrel{\text{(a)}}{=} \sum_{i=1}^n [H(Y_i|X_i) - H(Y_i|X_i,X^{i-1},M,Y^{i-1},Y_{i+1}^n)] \\
&= \sum_{i=1}^n [H(Y_i|X_i) - H(Y_i|X_i,U_i)], \text{ where }U_i=(X^{i-1},M,Y^{i-1},Y_{i+1}^n) \\
&= \sum_{i=1}^n I(U_i;Y_i|X_i) \\
&\stackrel{\text{(b)}}{=} n \sum_{i=1}^n \frac{1}{n} I(U_i;Y_i|X_i,T=i) \\
&= n\cdot I(U_T;Y_T|X_T,T) \\
&= n\cdot I(\underbrace{U_T,T}_{=\ U};Y_T|X_T) \\
&= n\cdot I(U;Y_T|X_T)
\end{align*}
In (a) we used the Markov chain $Y_i - (M,X^i,Y^{i-1})-Y_{i+1}^n$. As stated earlier, the random variable $T$ in (b) is independent of $(M,X^n,Y^n,Z^n,\hat{Z}^n)$ and uniformly distributed in $\{1,2,\hdots,n\}$. \\
}

{\it Single-letterizing \eqref{eq:1r-producing-output}:}
{\allowdisplaybreaks
\begin{align*}
0 &= I(\hat{Z}^n;X^n|M,Y^n) \\
&= H(X^n|M,Y^n) - H(X^n|M,Y^n,\hat{Z}^n) \\
&= \sum_{i=1}^n [H(X_i|X^{i-1},M,Y^n) - H(X_i|X^{i-1},M,Y^n,\hat{Z}^n)] \\
&\geq \sum_{i=1}^n [H(X_i|X^{i-1},M,Y^n) - H(X_i|X^{i-1},M,Y^n,\hat{Z}_i)] \\
&= \sum_{i=1}^n I(X_i;\hat{Z}_i|\underbrace{X^{i-1},M,Y^{i-1},Y_{i+1}^n}_{=\ U_i},Y_i) \\
&= \sum_{i=1}^n I(X_i;\hat{Z}_i|U_i,Y_i) \\
&= n\cdot I(X_T;\hat{Z}_T|U,Y_T)
\end{align*}
Last equality follows from the similar reasoning as before. \\
}

{\it Single-letterizing \eqref{eq:1r-privacy}:}
{\allowdisplaybreaks
\begin{align*}
n\epsilon &\geq I(M;X^n|Y^n,\hat{Z}^n) \\
&= H(X^n|Y^n,\hat{Z}^n) - H(X^n|M,Y^n,\hat{Z}^n) \\
&\stackrel{\text{(a)}}{\geq} H(X^n|Y^n,Z^n) - n\epsilon' - H(X^n|M,Y^n,\hat{Z}^n) \\
&= \sum_{i=1}^n [H(X_i|Y_i,Z_i) - H(X_i|X^{i-1},M,Y^n,\hat{Z}^n)] - n\epsilon' \\
&\geq \sum_{i=1}^n [H(X_i|Y_i,Z_i) - H(X_i|X^{i-1},M,Y^n,\hat{Z}_i)] - n\epsilon' \\
&\stackrel{\text{(b)}}{\geq} \sum_{i=1}^n [H(X_i|Y_i,\hat{Z}_i) -\epsilon'' - H(X_i|\underbrace{X^{i-1},M,Y^{i-1},Y_{i+1}^n}_{=\ U_i},Y_i,\hat{Z}_i)] - n\epsilon' \\
&= \sum_{i=1}^n [I(U_i;X_i|Y_i,\hat{Z}_i) - \epsilon' - \epsilon''] \\
&= \sum_{i=1}^n [I(U_i;X_i|Y_i,\hat{Z}_i) - \epsilon' - \epsilon''] \\
&\stackrel{\text{(c)}}{=} n[I(U_T;X_T|Y_T,\hat{Z}_T,T) -\epsilon' - \epsilon''] \\
&= n[I(U_T,T;X_T|Y_T,\hat{Z}_T) - I(T;X_T|Y_T,\hat{Z}_T) -\epsilon' - \epsilon''] \\
&\stackrel{\text{(d)}}{\geq} n[I(U_T,T;X_T|Y_T,\hat{Z}_T) -\epsilon''' - \epsilon' - \epsilon'']. \\
\end{align*}
We used the following fact in (a),(b), and (d): if $V$ and $V'$ are two random variables taking values in the same alphabet $\mathcal{V}$ such that $\|p_V-p_{V'}\|_1\leq\epsilon\leq1/4$, then it follows from \cite[Lemma 2.7]{CsiszarKorner11} that $|H(V)-H(V')|\leq \eta\log|\mathcal{V}|$, where $\eta\to0$ as $\epsilon\to0$.
Now \eqref{eq:1r-correctness} implies (a), i.e., $H(X^n|Y^n,\hat{Z}^n)\geq H(X^n|Y^n,Z^n)-n\epsilon'$, where $\epsilon'\to0$ as $\epsilon\to0$.
Note that \eqref{eq:1r-correctness} implies $\|p_{X_iY_i\hat{Z}_i}-p_{X_iY_iZ_i}\|_1\leq\epsilon$, for every $i\in[n]$, which implies (b), i.e., $H(X_i|Y_i,\hat{Z}_i)\geq H(X_i|Y_i,Z_i)-\epsilon_{i}''$, $i\in[n]$, where $\epsilon_{i}''\to0$ as $\epsilon\to0$; take $\epsilon''=\max_{i\in[n]}\{\epsilon_{i}''\}$.
(c) follows from the similar reasoning as before, where the random variable $T$ is independent of $(M,X^n,Y^n,Z^n,\hat{Z}^n)$ and is uniformly distributed in $\{1,2,\hdots,n\}$.
Inequality (d), i.e., $I(T;X_T|Y_T,\hat{Z}_T)\leq\epsilon'''$, with $\epsilon'''\to0$ as $\epsilon\to0$, can be shown along the lines of the proof of \Claimref{converse-small-claim}. So we get $I(U;X_T|Y_T,\hat{Z}_T)\leq \delta'$, where $\delta'=\epsilon+\epsilon'+\epsilon''+\epsilon'''$ and $\delta'\to0$ as $\epsilon\to0$. \\
}

{\it Single-letterizing \eqref{eq:1r-correctness}:} Below we prove that $\|p_{X^nY^n\hat{Z}^n}-p_{X^nY^nZ^n}\|_1\leq \epsilon \implies \|p_{X_iY_i\hat{Z}_i}-p_{X_iY_iZ_i}\|_1\leq \epsilon$, for every $i\in[n]$. This implies that 
\begin{align}
\|p_{X_TY_T\hat{Z}_T}-p_{X_TY_TZ_T}\|_1\leq \epsilon \text{ for a random }T. \label{eq:single_correctness}
\end{align}
Consider an arbitrary $i\in[n]$.
{\allowdisplaybreaks
\begin{align}
&\|p_{X_iY_i\hat{Z}_i}-p_{X_iY_iZ_i}\|_1 \notag \\
&= \sum_{x,y,z}|p_{X_iY_i\hat{Z}_i}(x,y,z)-p_{X_iY_iZ_i}(x,y,z)| \notag \\
&= \sum_{x,y,z}|\displaystyle \sum_{\substack{x^n,y^n,z^n:\\x_i=x,\\y_i=y,\\z_i=z}}p_{X^nY^n\hat{Z}^n}(x^n,y^n,z^n)-p_{X^nY^nZ^n}(x^n,y^n,z^n)| \notag \\
&\stackrel{\text{(a)}}{\leq} \sum_{x^n,y^n,z^n}|p_{X^nY^n\hat{Z}^n}(x^n,y^n,z^n)-p_{X^nY^nZ^n}(x^n,y^n,z^n)| \notag \\
&= \|p_{X^nY^n\hat{Z}^n}-p_{X^nY^nZ^n}\|_1 \notag \\
&\leq \epsilon, \label{eq:single_correctness-i}
\end{align}
where (a) follows from the triangle inequality. \\
}

\begin{proof}[Explanation of \eqref{eq:1r-rate-interim6}] 
For a fixed $(p_{XY},p_{Z|XY})$, let us define the following function:
\begin{align*}
R_{\epsilon}(p_{XY},p_{Z|XY}) := \min_{\substack{p_{U\hat{Z}|XY}: \\ I(U;Y|X)\leq\epsilon \\ I(U;X|Y,\hat{Z})\leq\epsilon \\ I(\hat{Z};X|U,Y)\leq\epsilon \\ \|p_{XY\hat{Z}}-p_{XYZ}\|_1\leq\epsilon}} I(X,\hat{Z};U|Y).
\end{align*}
Note that the expression in \eqref{eq:1r-rate-interim5} can be lower-bounded by $R_{\epsilon}(p_{XY},p_{Z|XY})-\delta$, because we relax two of the Markov chains in the definition of $R_{\epsilon}$. 
To prove the inequality in \eqref{eq:1r-rate-interim6}, it suffices to show that the function $R_{\epsilon}(p_{XY},p_{Z|XY})$ is right continuous at $\epsilon=0$. This can be proved (using continuity of mutual information and continuity of $L_1$-norm) along the lines of a similar inequality (inequality (14)) in \cite{Data16}.
\end{proof}

\section{Details omitted from \Subsectionref{full-support}}\label{app:characterization}
\begin{proof}[Proof of \Claimref{disjoint-messages}]
We prove this by contradiction. Suppose $\vec{\alpha}_i^{(y)}\neq\vec{\alpha}_j^{(y')}$ and $\U_{i}^{(y)}\cap\U_{j}^{(y')}\neq\phi$. Let $u\in\U_{i}^{(y)}\cap\U_{j}^{(y')}$.
Define two sets $Supp(\vec{\alpha}_i^{(y)})=\{x\in\X: \vec{\alpha}_i^{(y)}(x)>0\}$ and $Supp(\vec{\alpha}_j^{(y')})=\{x\in\X: \vec{\alpha}_j^{(y')}(x)>0\}$. We analyze two cases, one when these two sets are equal, and the other, when they are not.

{\bf Case 1.} $Supp(\vec{\alpha}_i^{(y)})=Supp(\vec{\alpha}_j^{(y')})$: 
since $\vec{\alpha}_i^{(y)}$ and $\vec{\alpha}_j^{(y')}$ are distinct probability vectors with same support, there must exist $x$ and $x'$ such that $\vec{\alpha}_i^{(y)}(x)>0$, $\vec{\alpha}_j^{(y')}(x)>0$, $\vec{\alpha}_i^{(y)}(x')>0$, $\vec{\alpha}_j^{(y')}(x')>0$, and $\frac{\vec{\alpha}_i^{(y)}(x')}{\vec{\alpha}_i^{(y)}(x)}\neq\frac{\vec{\alpha}_j^{(y')}(x')}{\vec{\alpha}_j^{(y')}(x)}$.
Since $\vec{\alpha}_i^{(y)}(x)>0, \vec{\alpha}_i^{(y)}(x')>0$, we have $p_{Z|XY}(z|x,y)>0$ and $p_{Z|XY}(z|x',y)>0$ for every $z\in\Z_{i}^{(y)}$ . Similarly, since $\vec{\alpha}_j^{(y')}(x)>0,\vec{\alpha}_j^{(y')}(x')>0$, we have $p_{Z|XY}(z'|x,y')>0$ and $p_{Z|XY}(z'|x',y')>0$ for every $z'\in\Z_{j}^{(y')}$.
Note that $\frac{\vec{\alpha}_i^{(y)}(x')}{\vec{\alpha}_i^{(y)}(x)}=\frac{p_{Z|XY}(z|x',y)}{p_{Z|XY}(z|x,y)}$ and $\frac{\vec{\alpha}_j^{(y')}(x')}{\vec{\alpha}_j^{(y')}(x)}=\frac{p_{Z|XY}(z'|x',y')}{p_{Z|XY}(z'|x,y')}$, where, by hypothesis, $\frac{p_{Z|XY}(z|x',y)}{p_{Z|XY}(z|x,y)}\neq\frac{p_{Z|XY}(z'|x',y')}{p_{Z|XY}(z'|x,y')}$.

Since $u\in\U_{i}^{(y)}$, there exists $z\in\Z_{i}^{(y)}$ such that $p(u,z|y)>0$. This implies -- by privacy against Bob -- that $p(u,z|x,y)>0$ and $p(u,z|x',y)>0$.
Consider $p_{UZ|XY}(u,z|x,y)$ and expand it as follows:
\begin{align}
p_{UZ|XY}(u,z|x,y) &= p_{Z|XY}(z|x,y)p_{U|XYZ}(u|x,y,z) \notag \\
&= p_{Z|XY}(z|x,y)p_{U|YZ}(u|y,z) \label{eq:disjoint-messages-interim1}
\end{align}
Since $p_{UZ|XY}(u,z|x,y)>0$, all the terms above are non-zero and well-defined. In the last equality we used privacy against Bob to write $p_{U|XYZ}(u|x,y,z)=p_{U|YZ}(u|y,z)$. Now expand $p_{UZ|XY}(u,z|x,y)$ in another way as follows:
\begin{align}
p_{UZ|XY}(u,z|x,y) &= p_{U|XY}(u|x,y)p_{Z|UXY}(z|u,x,y) \notag \\
&= p_{U|XY}(u|x)p_{Z|UY}(z|u,y) \label{eq:disjoint-messages-interim2}
\end{align}
Again, all the terms above are non-zero and well-defined because $p_{UZ|XY}(u,z|x,y)>0$. We used $U-X-Y$ to write $p_{U|XY}(u|x,y)=p_{U|X}(u|x)$ and $Z-(U,Y)-X$ to write $p_{Z|UXY}(z|u,x,y)=p_{Z|UY}(z|u,y)$ in \eqref{eq:disjoint-messages-interim2}. Now comparing \eqref{eq:disjoint-messages-interim1} and \eqref{eq:disjoint-messages-interim2} gives the following:
\begin{align}
p_{Z|XY}(z|x,y)p_{U|YZ}(u|y,z)=p_{U|X}(u|x)p_{Z|UY}(z|u,y) \label{eq:disjoint-messages-interim3}
\end{align}
Since $p_{UZ|XY}(u,z|x',y)>0$, we can apply the same arguments as above with $p_{UZ|XY}(u,z|x',y)$ and get the following:
\begin{align}
p_{Z|XY}(z|x',y)p_{U|YZ}(u|y,z)=p_{U|X}(u|x')p_{Z|UY}(z|u,y) \label{eq:disjoint-messages-interim4}
\end{align}
Note that all the terms on both sides of \eqref{eq:disjoint-messages-interim3} and \eqref{eq:disjoint-messages-interim4} are non-zero. Dividing \eqref{eq:disjoint-messages-interim3} by \eqref{eq:disjoint-messages-interim4} gives the following:
\begin{align}
\frac{p_{Z|XY}(z|x,y)}{p_{Z|XY}(z|x',y)} = \frac{p_{U|X}(u|x)}{p_{U|X}(u|x')}. \label{eq:first-ratio}
\end{align}
Similarly, since $u\in\U_{j}^{(y')}$, there exists $z'\in\Z_{j}^{(y')}$ such that $p(u,z'|y')>0$. This implies -- by privacy against Bob -- that $p(u,z'|x,y')>0$ and $p(u,z'|x',y')>0$. Applying the above arguments with $p(u,z'|x,y')>0$ and $p(u,z'|x',y')>0$ gives
\begin{align}
\frac{p_{Z|XY}(z'|x,y')}{p_{Z|XY}(z'|x',y')}=\frac{p_{U|X}(u|x)}{p_{U|X}(u|x')}. \label{eq:second-ratio}
\end{align}
Comparing \eqref{eq:first-ratio} and \eqref{eq:second-ratio} gives $\frac{p_{Z|XY}(z|x,y)}{p_{Z|XY}(z|x',y)}=\frac{p_{Z|XY}(z'|x,y')}{p_{Z|XY}(z'|x',y')}$, which is a contradiction.\\
{\bf Case 2.} $Supp(\vec{\alpha}_i^{(y)})\neq Supp(\vec{\alpha}_j^{(y')})$: 
assume, without loss of generality, that $Supp(\vec{\alpha}_i^{(y)})\setminus Supp(\vec{\alpha}_j^{(y')})\neq\phi$. Let $x\in Supp(\vec{\alpha}_i^{(y)})\setminus Supp(\vec{\alpha}_j^{(y')})$. This implies that $\vec{\alpha}_i^{(y)}(x)>0$ and $\vec{\alpha}_j^{(y')}(x)=0$. Since $Supp(\vec{\alpha}_j^{(y')})\neq \phi$ (because $\vec{\alpha}_j^{(y')}$ is not a zero vector), there exists $x'\in Supp(\vec{\alpha}_j^{(y')})$ such that $\vec{\alpha}_j^{(y')}(x')>0$.
Note that $\vec{\alpha}_i^{(y)}(x)>0$ implies $p_{Z|XY}(z|x,y)>0$ for every $z\in\Z_{i}^{(y)}$; $\vec{\alpha}_j^{(y')}(x)=0$  implies $p_{Z|XY}(z'|x,y')=0$ for every $z'\in\Z_{j}^{(y')}$; and $\vec{\alpha}_j^{(y')}(x')>0$ implies $p_{Z|XY}(z'|x',y')>0$ for every $z'\in\Z_{j}^{(y')}$.
Since $u\in\U_{i}^{(y)}\cap\U_{j}^{(y')}$, there exists $z\in\Z_{i}^{(y)}$ and $z'\in\Z_{j}^{(y')}$ such that $p(u,z|y)>0$ and $p(u,z'|y')>0$. 
These imply -- by privacy against Bob -- that $p(u,z|x,y)>0$ and $p(u,z'|x',y')>0$.
Now consider $p(u,z'|x,y')$ and expand it as follows:
\begin{align}
p(u,z'|x,y') &= p(u|x,y')p(z'|x,y',u) \label{eq:second-comparison1} \\
&= p(u|x)p(z'|y',u) \label{eq:second-comparison2}
\end{align}
We used the Markov chain $U-X-Y$ to write $p(u|x,y')=p(u|x)$, where $p(u|x)>0$ because $p(u,z|x,y)>0$. We used the Markov chain $Z-(Y,U)-X$ to write $p(z'|x,y',u)=p(z'|y',u)$ in \eqref{eq:second-comparison1}, which is greater than zero because $p(u,z'|x',y')>0$. Putting all these together in \eqref{eq:second-comparison2} gives $p(u,z'|x,y')>0$, which implies $p_{Z|XY}(z'|x,y')>0$, a contradiction.
\end{proof}
\vspace{0.25cm}
\begin{proof}[Proof of \Claimref{equal-alpha-vectors}]
Since the Markov chain $U-X-Y$ holds, we have that the set of messages that Alice sends to Bob are same for all inputs of Bob. Now it follows from \Claimref{disjoint-messages} that for every $y\in\Y$ we can partition the set of all possible messages $\U$ as follows: $\U=\U_1^{(y)}\biguplus\U_2^{(y)}\biguplus\hdots\biguplus\U_{k(y)}^{(y)}$, where $\U_i^{(y)}$ is as defined earlier in \eqref{eq:one-round_set-msgs}. Consider any two $y,y'\in\Y$. We have
\begin{align}
\U=\U_1^{(y)}\biguplus\U_2^{(y)}\biguplus\hdots\biguplus\U_{k(y)}^{(y)}=\U_1^{(y')}\biguplus\U_2^{(y')}\biguplus\hdots\biguplus\U_{k(y')}^{(y')}.\label{eq:equal-partition_messages}
\end{align}
First observe that $k(y)=k(y')$. Otherwise, there exists $i\in[k(y)]$ and $j\in [k(y')]$ such that $\U_i^{(y)}\cap \U_j^{(y')}\neq\phi$ and $\vec{\alpha}_i^{(y)}\neq \vec{\alpha}_j^{(y')}$, which contradicts \Claimref{disjoint-messages}. From now on we denote $k(y)$ by $k$ for every $y\in\Y$.

Suppose $\{\vec{\alpha}_1^{(y)},\vec{\alpha}_2^{(y)},\hdots,\vec{\alpha}_{k}^{(y)}\}\neq\{\vec{\alpha}_1^{(y')},\vec{\alpha}_2^{(y')},\hdots,\vec{\alpha}_{k}^{(y')}\}$. It means that there exists an $i\in [k]$ such that $\vec{\alpha}_i^{(y)}\notin\{\vec{\alpha}_1^{(y')},\vec{\alpha}_2^{(y')},\hdots,\vec{\alpha}_{k}^{(y')}\}$. Since $\U_i^{(y)}$ is associated with $\vec{\alpha}_i^{(y)}$, we have by \eqref{eq:equal-partition_messages} that $\U_i^{(y)}\notin\U_1^{(y')}\biguplus\U_2^{(y')}\biguplus\hdots\biguplus\U_{k}^{(y')}$. This contradicts \eqref{eq:equal-partition_messages}.
\end{proof}

\vspace{0.25cm}
\begin{proof}[Proof of \Theoremref{one-round-characterization}]
$\Rightarrow$: (1) has been shown in \Claimref{equal-alpha-vectors}. (2) follows from the following argument:
take any $x\in\X$ and consider the following set of equalities:
\begin{align*}
\sum_{z\in\Z_{i}^{(y)}}p_{Z|XY}(z|x,y) \stackrel{\text{(a)}}{=} \sum_{u\in\U_i} p_{U|XY}(u|x,y) 
\stackrel{\text{(b)}}{=} \sum_{u\in\U_i} p_{U|XY}(u|x,y') 
\stackrel{\text{(c)}}{=} \sum_{z\in\Z_{i}^{(y')}}p_{Z|XY}(z|x,y'),
\end{align*}
where (a) and (c) follow from the fact that Bob's output $Z$ belongs to $\Z_{i}^{(y)}$, when his input is $y$ (or belongs to $\Z_{i}^{(y')}$, when his input is $y'$) if and only if the message $U$ that Alice sends to Bob belongs to $\U_{i}$. (b) follows from the Markov chain $U-X-Y$. \\

$\Leftarrow$: we show this direction by giving a secure protocol in \Figureref{one-round-protocol}. Now we prove that this protocol is perfectly secure, i.e., it satisfies perfect correctness and perfect privacy. \\

\begin{figure}
\hrule
\vspace{0.20cm}
{\bf One-round secure protocol}
\vspace{0.20cm}
\hrule
\vspace{0.25cm}
\noindent {\bf Input:} Alice has $x\in\X$ and Bob has $y\in\Y$. \\
\noindent {\bf Output:} Bob outputs $z$ with probability $p_{Z|XY}(z|x,y)$.
\begin{center}{\bf Protocol}\end{center}
\begin{enumerate}
\item Both Alice and Bob agree on an element in the alphabet $\Y$, say $y_1$, beforehand.
\item Alice sends $u_i$ to Bob with probability $\sum_{z\in\Z_i^{(y)}}p_{Z|XY}(z|x,y_1)$.
\item Upon receiving $u_i$, Bob fixes any element $x'$ for which $\alpha_i^{(y)}(x')>0$, and outputs $z$ with probability
$\frac{p_{Z|XY}(z|x',y)}{\sum_{z\in\Z_{i}^{(y)}}p_{Z|XY}(z|x',y)}$.
\end{enumerate}
\hrule
\vspace{0.20cm}
\caption{A one-round secure protocol.}
\label{fig:one-round-protocol}
\end{figure}
{Correctness:} suppose the protocol of \Figureref{one-round-protocol} produces an output according to the p.m.f. $p(z|x,y)$. We show below that $p(z|xy)$ is equal to $p_{Z|XY}(z|x,y)$.
\begin{align}
p(z|x,y) &= p(z,u_i|x,y) \label{eq:char_corr-interim1} \\
&= p(u_i|x,y)\times p(z|u_i,x,y) \notag \\
&= p(u_i|x)\times p(z|u_i,y) \label{eq:char_corr-interim2} \\
&= \Big(\sum_{z\in\Z_i^{(y)}}p_{Z|XY}(z|x,y_1)\Big)\times\Big(\frac{p_{Z|XY}(z|x',y)}{\sum_{z\in\Z_{i}^{(y)}}p_{Z|XY}(z|x',y)}\Big) \label{eq:char_corr-interim3} \\
&= \Big(\sum_{z\in\Z_i^{(y)}}p_{Z|XY}(z|x,y)\Big)\times\Big(\frac{p_{Z|XY}(z|x,y)}{\sum_{z\in\Z_{i}^{(y)}}p_{Z|XY}(z|x,y)}\Big) \label{eq:char_corr-interim4} \\
&= p_{Z|XY}(z|x,y) \notag
\end{align}
In \eqref{eq:char_corr-interim1} we assume that $z\in\Z_i^{(y)}$; and \eqref{eq:char_corr-interim1} is an equality because the message that Alice sends to Bob is a deterministic function of Bob's input and output. In \eqref{eq:char_corr-interim2} we used the Markov chain $U-X-Y$ to write $p(u_i|x,y)=p(u_i|x)$ and $Z-(U,Y)-X$ to write $p(z|u_i,x,y)=p(z|u_i,y)$. In \eqref{eq:char_corr-interim3} we substituted the values of $p(u_i|x,y)$ and $p(z|u_i,y)$ from the protocol of \Figureref{one-round-protocol}. In \eqref{eq:char_corr-interim4} we used the facts that $\sum_{z\in\Z_i^{(y)}}p_{Z|XY}(z|x,y_1)=\sum_{z\in\Z_i^{(y)}}p_{Z|XY}(z|x,y)$ (which follows from the assumption -- see the second item in the theorem statement) and $\frac{p_{Z|XY}(z|x',y)}{\sum_{z\in\Z_{i}^{(y)}}p_{Z|XY}(z|x',y)}=\frac{p_{Z|XY}(z|x,y)}{\sum_{z\in\Z_{i}^{(y)}}p_{Z|XY}(z|x,y)}$ (which follows from the fact that the matrix $A_i^{(y)}=\vec{\alpha}_i^{(y)}\times\vec{\gamma}_i^{(y)}$, defined earlier, is a rank-one matrix). \\

{Privacy:} we need to show that if $p_{Z|XY}(z|x_1,y)>0$ and $p_{Z|XY}(z|x_2,y)>0$ for some $x_1,x_2,y,z$, then for every $u_i$, $p(u_i|x_1,y,z)=p(u_i|x_2,y,z)$ must hold. This follows because $p(u_i|x,y,z)=\mathbbm{1}_{\{z\in\Z_i^{(y)}\}}$, i.e., $u_i$ is a deterministic function of $(y,z)$ and is independent of Alice's input.
\end{proof}
\vspace{0.25cm}

\begin{proof}[Proof of \Theoremref{rate_full-support}]
\Lemmaref{as-iff-ps} states that $(p_{XY},p_{Z|XY})$ is computable with asymptotic security if and only if it is computable with perfect security, i.e., $(p_{XY},p_{Z|XY})$ is computable with asymptotic security if and only if there exists $p(u,x,y,z)=p_{XYZ}(x,y,z)p(u|x,y,z)$ that satisfies \eqref{eq:ps_correctness}-\eqref{eq:ps_producing-output}). To prove this theorem, we use the characterization of securely computable $(p_{XY},p_{Z|XY})$ with perfect security (where $p_{XY}$ has full support) from \Theoremref{one-round-characterization}. First we give achievability and then prove converse. Although the general achievability of \Theoremref{as_rate} is applicable here, we give a more direct achievability proof based on the Slepian-Wolf coding scheme, with a much simpler argument of its correctness and privacy. \\

{\bf Achievability:} Alice and Bob get $X^n$ and $Y^n$, respectively, as their inputs, where $(X_i,Y_i)$ are distributed i.i.d. according to $p_{XY}$. Alice passes $X^n$ through the virtual DMC $p_{W|X}$ and obtains an i.i.d. sequence $W^n$, where $p_{W|X}$ is obtained from \eqref{eq:W-interim2}. Note that if Bob gets access to this $W^n$ sequence exactly, and he outputs an i.i.d. sequence $Z^n$ by passing $(W^n,Y^n)$ through the virtual DMC $p_{Z|WY}$, where $p_{Z|WY}$ is obtained from \eqref{eq:W-interim2}, then perfect correctness is achieved, i.e., Bob's output $Z^n$ will be distributed exactly according to $\Pi_{i=1}^np_{Z|XY}(z_i|x_i,y_i)$. We show below that conveying $W^n$ to Bob also maintains privacy against Bob.
So, the goal is for Alice to convey this $W^n$ sequence to Bob reliably. Since $W^n$ may be correlated with $Y^n$, Alice can use the Slepian-Wolf coding scheme to convey $W^n$ to Bob at a rate $H(W|Y)$. Now we show that this scheme is asymptotically secure, i.e., it satisfies asymptotic correctness \eqref{eq:correctness_S} and asymptotic privacy \eqref{eq:privacy-bob}. In the analysis below, $E$ is an indicator random variable which is equal to 1 if Bob recovers the intended $W^n$ exactly, and 0 otherwise.  Let $\epsilon':=\Pr\{E=1\}$, then by the correctness of Slepian-Wolf coding scheme we have that $\epsilon'\to0$ as $n\to\infty$. Note that if $E=0$, i.e., Bob recovers the intended $W^n$ exactly, then $p_{\hat{Z}^n|X^nY^n}=p_{Z^n|X^nY^n}$. \\

{Correctness:} Note that Bob makes an error only when he does not decode the intended $W^n$ sequence exactly; and that happens with probability $\epsilon'$, which goes to 0 as $n$ tends to infinity. This implies that our scheme satisfies $\|p_{X^nY^n\hat{Z}^n}-p_{X^nY^nZ^n}\|_1\to0$ as $n\to\infty$. \\

{Privacy:} We show that our scheme satisfies $I(M;X^n|Y^n,\hat{Z}^n)\leq n\epsilon''$, where $\epsilon''\to0$ as $n\to\infty$.
\begin{align}
I(M;X^n|Y^n,\hat{Z}^n) &= I(M;X^n|Y^n,\hat{Z}^n,\hat{W}^n(Y^n,\hat{Z}^n)) \label{eq:achieve-privacy-interim1} \\
&\leq I(M;X^n|Y^n,\hat{Z}^n,\hat{W}^n(Y^n,\hat{Z}^n),E) + H(E) \label{eq:achieve-privacy-interim2} \\
&= \epsilon'\times I(M;X^n|Y^n,\hat{Z}^n,\hat{W}^n(Y^n,\hat{Z}^n),E=1) \notag \\ 
&\qquad + (1-\epsilon')\times \underbrace{I(M;X^n|Y^n,\hat{Z}^n,\hat{W}^n(Y^n,\hat{Z}^n),E=0)}_{=\ I(M;X^n|Y^n,\hat{Z}^n,W^n,E=0)\ =\ 0} + H_2(\epsilon') \label{eq:achieve-privacy-interim3} \\
&\leq \epsilon'\times n \times \log|\X| + H_2(\epsilon') \label{eq:achieve-privacy-interim4} \\
&\leq n\epsilon'', \qquad \text{ where }\epsilon''\to0 \text{ as } n\to\infty.  \notag
\end{align}
In our achievability scheme, Bob first decodes $\hat{W}^n$ from $(M,Y^n)$ and then samples $\hat{Z}^n$. Note that the same $\hat{W}^n$ can be recovered from $(Y^n,\hat{Z}^n)$. Let $\hat{W}:\Y\times\hat{Z}\to[k]$ -- note that $[k]$ is the alphabet of $W$ -- be the function by which we recover $\hat{W}$ from $(Y,\hat{Z})$.
In \eqref{eq:achieve-privacy-interim1}, the $\hat{W}^n(Y^n,\hat{Z}^n)$ vector is equal to the $\hat{W}^n$ vector that Bob decodes from $(Y^n,\hat{Z}^n)$. In \eqref{eq:achieve-privacy-interim2}, we used the following information inequality: $I(A;B|C)\leq I(A;B|C,D) + H(D)$ for any random variables $A,B,C,D$, which can be proved as follows:
\begin{align*}
I(A;B|C) &\leq I(A;B,D|C) \\
&\leq I(A;D|C)+I(A;B|C,D) \\
&\leq H(D) + I(A;B|C,D).
\end{align*}
In \eqref{eq:achieve-privacy-interim3}, $I(M;X^n|Y^n,\hat{Z}^n,W^n,E=0)=0$, because $M$ is a determinisitc function of $W^n$; and $H_2(.)$ denotes the binary entropy function. In \eqref{eq:achieve-privacy-interim4} we used $I(M;X^n|Y^n,\hat{Z}^n,\hat{W}^n(Y^n,\hat{Z}^n),E=1)\leq \log |\X|^n$. Since $\X$ is finite, $\epsilon''\to0$ as $\epsilon'\to0$. \\

{\bf Converse:}
If $p_{XY}$ has full support, then we can simplify the expression in \eqref{eq:1r-rate-interim6}. Note that $p_{UZ|XY}$ in \eqref{eq:1r-rate-interim6} satisfies \eqref{eq:ps_correctness}-\eqref{eq:ps_privacy}, which means that $p_{UZ|XY}$ computes $(p_{XY},p_{Z|XY})$ with perfect security.
This implies that the random variable $W$ is a function of $U$ as well as of $(Y,Z)$ (see the comments before \Theoremref{rate_full-support}). Now, we can lower-bound the objective function in \eqref{eq:1r-rate-interim6} by $H(W|Y)$ as follows: $I(Z;U|Y)= I(Z,W;U,W|Y)\geq H(W|Y)$. Note that $U=W$ satisfies all the Markov chains in \eqref{eq:1r-rate-interim6}. Thus, if $p_{XY}$ has full support, then we can identify $U$ as $W$.
\end{proof}

\section{Proof of \Theoremref{rate_no-privacy}}\label{app:proof_no-privacy}
As stated in the beginning of \Sectionref{proof_general-case}, the cardinality bound of $|\U|\leq|\X|\cdot|\Y|\cdot|\Z|+2$ follows from the Fenchel-Eggleston's strengthening of Carath\'eodory's theorem \cite[pg. 310]{CsiszarKorner11}. \\

{\bf Converse:} For every $\epsilon>0$, there is a large enough $n$ and a code $\C_n$ that satisfies the following properties:
\begin{align}
& M - X^n - Y^n, \notag \\
& \hat{Z}^n - (M,Y^n) - X^n, \notag \\
& \|p_{X^nY^n\hat{Z}^n}-p_{X^nY^nZ^n}\|_1\leq \epsilon. \notag
\end{align}
Now we follow the converse of \Theoremref{as_rate} exactly (except for the privacy part) up to \eqref{eq:1r-rate-interim6}, which gives the following:
\begin{align*}
R&\geq \displaystyle \min_{\substack{p_{U|XYZ}: \\ U-X-Y \\ Z-(U,Y)-X \\}} I(X,Z;U|Y) - \epsilon_4 - \delta,
\end{align*}
where $\epsilon_4+\delta\to0$ as $n\to\infty$. \\

{\bf Achievability:} This is simpler than the achievability of \Theoremref{as_rate} as we do not need to worry about privacy here. We follow the proof of achievability of \Theoremref{as_rate} (given in \Appendixref{achievability}) with two modifications: (i) $p_{UXYZ}$ does not need to satisfy the Markov chain $U-(Y,Z)-X$ in protocol A, and (ii) we do not need to show that $\hat{p}(x^n,y^n,z^n,m|f)$ satisfies privacy condition in part (3). With this we have the following constraints on rates $R'$ and $R_{M}$ (similar to \eqref{eq:part2-constraint2} and \eqref{eq:part2-constraint3}):
\begin{align}
R' &< H(U|X,Y,Z), \label{eq:f-rate_noPrivacy} \\
R'+R_{M} &> H(U|Y). \label{eq:sum-rate_noPrivacy}
\end{align} 
Now it follows that, for every $R_M>I(X,Z;U|Y)$, there exists $R'>0$ such that $R'<H(U|X,Y,Z)$ and $R_M+R'>H(U|Y)$ hold, which implies existence of a coding scheme by the above analysis with rate $R_M$.
We can define a set of appropriate encoder and decoder similar to how we defined them at the end of the achievability proof of \Theoremref{as_rate}.

\end{document}